\documentclass[runningheads]{llncs}

\usepackage{cite}
\usepackage{pifont}
\usepackage{amssymb}
\usepackage{comment}

\usepackage{mathtools}
\usepackage{algorithm}
\usepackage{algorithmicx}
\usepackage{algpseudocodex} 
\usepackage{subcaption}

\usepackage{xcolor}
\usepackage{url}
\definecolor{okabe1}{HTML}{000000}
\definecolor{okabe2}{HTML}{E69F00}
\definecolor{okabe3}{HTML}{56B4E9}
\definecolor{okabe4}{HTML}{009E73}
\definecolor{okabe5}{HTML}{F0E442}
\definecolor{okabe6}{HTML}{0072B2}
\definecolor{okabe7}{HTML}{D55E00}
\definecolor{okabe8}{HTML}{CC79A7}

\usepackage{hyperref}
\usepackage{cleveref}

\spnewtheorem{observation}{Observation}{\bfseries}{\itshape}
\crefname{observation}{Observation}{Observations}
\Crefname{observation}{Observation}{Observations}

\renewcommand{\emph}[1]{\textbf{\textit{#1}}}
\usepackage[]{todonotes} 

\title{How to Sort in a Refrigerator:\\ Simple Entropy-Sensitive Strictly In-Place Sorting Algorithms}
\titlerunning{Entropy-Sensitive Strictly In-Place Sorting}

\author{Ofek Gila\inst{}\orcidID{0009-0005-5931-771X} \and
Michael T. Goodrich\inst{}\orcidID{0000-0002-8943-191X} \and
Vinesh~Sridhar\inst{}\orcidID{0009-0009-3549-9589}}

\authorrunning{O. Gila, M.T. Goodrich, and V. Sridhar}

\institute{University of California, Irvine, USA\\
\email{\{ogila,goodrich,vineshs1\}@uci.edu}}

\date{}

\begin{document}

\maketitle
\begin{abstract}
While modern general-purpose computing systems have ample amounts of memory, it is still the case 
that embedded computer systems, such as in a refrigerator, are memory limited; hence, such embedded systems
motivate the need for strictly in-place algorithms, which use only $O(1)$ additional memory besides that used for the input.
In this paper,
we provide the first comparison-based sorting algorithms
that are strictly in-place and have a running time that is optimal in terms of the run-based entropy, $\mathcal{H}(A)$, of an input array, $A$, of size~$n$.
In particular, we describe two remarkably 
simple paradigms for implementing stack-based 
natural mergesort algorithms to be strictly in-place in $O(n(1+\mathcal{H}(A)))$ time.
\end{abstract}

\section{Introduction}

Embedded computer systems are designed for a specific function within a larger device, and they comprise a processor
and a limited amount of memory. They are found in diverse applications such as household appliances, vehicles, 
and medical equipment. Unlike with general-purpose computers, it is often difficult to extend the limited
memory in an embedded systems; hence, embedded systems are a natural application domain for strictly in-place
algorithms, which use only $O(1)$ additional memory besides that used for the input.
Further, embedded systems that utilize non-volatile main memory (NVMM)~\cite{nvmm} have additional need of
strictly in-place algorithms, because these memories degrade if they are over-written many times, such as if
they had to maintain a data structure such as a stack.
Still, we would like to be able to implement fast algorithms in embedded systems, since they
often have real-time performance requirements.

For example, in
\emph{beyond worst-case analysis of algorithms}~\cite{roughgarden2019beyond,roughgarden2021beyond},
we are interested in algorithms whose running time can beat worst-case bounds depending on parameters or metrics defined
for problem instances that capture typical properties for expected inputs.
For example, suppose we are given an array,
$A=[x_1,x_2,\ldots,x_n]$, of elements that 
come from a total order.
If we let $\mathcal{R}=\{R_1,R_2,\ldots,R_{\rho(A)}\}$ 
be a partition of~$A$ into a set of maximal increasing or decreasing \emph{runs}
(i.e., subsequences of consecutive elements), and we
let $r_i$ denote the size, $|R_i|$, of the $i$-th run in $A$
(so $\sum_{i=1}^{\rho(A)} r_i \,=\, n$), then
the \emph{run-based entropy}, $\mathcal{H}({A})$,
for $A$ is defined as 
follows:\footnote{All logarithms used 
               in this paper are base $2$ unless otherwise stated.}
\[
\mathcal{H}({A}) = \sum_{i=1}^{\rho(A)} (r_i/n) \log (n/r_i).
\]
As it turns out, researchers have proven that any comparison-based sorting
algorithm has a lower bound for its running time
that is $\Omega(n(1+\mathcal{H}({A})))$, and there are several sorting algorithms
that run in $O(n(1+\mathcal{H}({A})))$ time;
see, e.g.,~\cite{munro18,juge24,
gelling,auger2019,buss19,
Takaoka09,schou2024persisort,BARBAY2013109}.
Interestingly, all of these algorithms are
stack-based natural mergesort algorithms, where we start from the runs and implement
a type of mergesort that uses a stack to determine which runs to merge.
In fact, two of these algorithms, TimSort and PowerSort, have
become widely used in practice, e.g., as the default sorting algorithms in
Python, Java, and other major programming languages.
As their description suggests, these algorithms maintain a stack of runs
and decide when to merge them based on their
lengths and/or position in order to take advantage of the
pre-sortedness of the input data. Unfortunately, none of these algorithms are strictly in-place, i.e., 
they use $\omega(1)$ additional memory cells besides
the $n$ cells used for the input array. 
In fact,
their stack
alone can be as large as $\Omega(\log n)$.
Thus, none of these algorithms are suitable for embedded systems where we wish to maintain
only an $O(1)$-sized amount of information in addition to the input, which is the subject of this paper.
Alternatively, none of the known strictly in-place sorting algorithms, such as heapsort, are entropy-sensitive.

\paragraph{Prior Related Work.}
The existing stack-based natural mergesort algorithms generally
consider the top $k$ runs on the
stack 
and merge adjacent runs (using the standard sorted-list merge algorithm)
according to some rule, where $k\ge 2$ is a constant that depends on the algorithm.
Munro and Wild~\cite{munro18} introduce 
PeekSort and PowerSort, which use rules based on powers of $2$,
and achieve an instance-optimal running time of 
$O(n(1+\mathcal{H}({A})))$.
Auger, Jug{\'e}, Nicaud, and Pivoteau~\cite{auger2015,auger2019} study
the popular TimSort algorithm~\cite{tim}, 
whose stack-based combination rule is
loosely based on the Fibonacci sequence, and show that
it is also instance-optimal.\footnote{TimSort was recently replaced
   by PowerSort in the CPython reference implementation 
   of Python~\protect\cite{gelling}.}
Other stack-based natural mergesort algorithms
with this running time
include the Adaptive ShiversSort of Jug{\'e}~\cite{juge24};
the Multiway PowerSort of
Gelling, Nebel, Smith, and Wild~\cite{gelling}; and the $\alpha$-MergeSort~\cite{buss19} of Buss and Knop. 
Takaoka’s MinimalSort~\cite{Takaoka09},
Schou and Wang's PersiSort~\cite{schou2024persisort},
and adaptive search-tree sorting by
Barbay and Navarro~\cite{BARBAY2013109} take different approaches yet still achieve instance-optimality.
Furthermore,
TimSort and PowerSort have become the reference sorting algorithms for popular software libraries, including for Python, Java, and Swift~\cite{gelling}.

There are two main challenges for implementing
stack-based  mergesort algorithms to be strictly in-place.
The first, which has been extensively studied, is how to merge two sorted subarrays in-place.
Fortunately, there are existing methods that perform linear-time merging strictly in-place; 
see, e.g.,~\cite{CHEN200634,DALKILIC20111049,CHEN2003191,ellis,GEFFERT2000159,huang1992fast,huang2,kat,wong1981some}.
Even with an in-place merge method, however,
we 
must overcome the second challenge to 
implementing stack-based mergesorts strictly in-place---\emph{the stack itself}---which can be as deep
as $\Omega(\log n)$.

\paragraph{Our Results.}
We define a simple 
\emph{walk-back algorithm}, which can implement
any stack-based mergesort algorithm to be strictly in-place.
In this algorithm, we only maintain
the lengths of a constant number of runs on the stack plus 
at most $O(1)$ additional metadata.
If we ever need to determine
the length of a run that occurs deeper in the stack,
we simply walk backwards down the array to recover it.   
Also, if the algorithm used to merge runs 
is stable~\cite{huang1992fast,DALKILIC20111049,CHEN2003191,ellis} 
(i.e., constructed so that equal elements retain their relative order 
from the input array, $A$), 
then the resulting in-place natural mergesort is also stable.

We call a stack-based natural mergesort algorithm, $\mathcal{S}$,
\emph{walkable} if applying the
walk-back algorithm to $\mathcal{S}$ increases the running time of $\mathcal{S}$ only
by a constant factor. 
Admittedly, this definition hides some important details about our analysis of the walk-back algorithm.
In \Cref{sec:walkable}, we give our main result, showing that PowerSort~\cite{munro18} is 
walkable. 
In Appendices~\ref{sec:shivers-appendix} and~\ref{sec:additional-proofs}, 
we show that several other stack-based mergesorts are walkable, 
including $c$-Adaptive ShiversSort, ShiversSort, $\alpha$-StackSort, and 2-MergeSort~\cite{juge24,buss19}. 
In \Cref{sec:counter}, we give a negative result and show that the well-known adaptive sorting algorithms TimSort\footnote
{We show that the original version of TimSort is walkable despite having a bug~\cite{de2015openjdk}. 
In fact, the bug concerns the failure to allocate enough memory to maintain the stack of runs.
See \Cref{sec:additional-proofs}.} ~\cite{tim,auger2015,auger2019}
and $\alpha$-MergeSort~\cite{buss19}
are \textit{not} walkable.
%
%

Our negative results motivate our second simple algorithm, which we call the 
\emph{jump-back algorithm}, that can also be used to implement
virtually all stack-based mergesorts in-place.
In this algorithm, we only maintain
the lengths of a constant number of runs on the stack plus 
at most $O(1)$ metadata, and we encode the lengths of
runs in-place in the runs themselves using simple bit-encoding methods.
If we ever need to determine
the length of a run that occurs deeper in the stack,
we decode its length from its bit encoding in $O(\log n)$ time, 
which then allows us to jump to the beginning of the run (and possibly
repeat this lookup for the next earlier run).
In \Cref{sec:encoding},
      we show that 
      we can implement any stack-based natural mergesort algorithm
      to be in-place using the jump-back algorithm while only increasing
      its running time by a small constant factor,
which is a property we call \emph{jumpable}.
Combining our main contributions with known results for 
the running times of
stack-based mergesort algorithms
and in-place merge algorithms
implies the first strictly in-place
sorting algorithms whose running times are 
instance-optimal with respect to run-based entropy. 

We support our theoretical results with experiments found in \Cref{sec:experiments}.
%
%


\section{Preliminaries}
Let us begin with some preliminaries for stack-based mergesort
algorithms.
Let $A$ be an input array of $n$ elements.
We define a run as either a non-decreasing or strictly decreasing subsequence of
elements in $A$.
All decreasing runs can easily be flipped in-place
to be increasing runs in an initial
linear-time sweep of $A$
without affecting stability, so, without loss of generality,
we only consider non-decreasing runs.

We can separate all stack-based mergesort algorithms into two phases.
First is the \emph{main loop}, in which we sweep $A$ from left-to-right, pushing
new runs into our stack.
Each algorithm defines
and is distinguished by
a merge policy, $\mathcal{P}$.
Whenever we push a new run into
our stack and $\mathcal{P}$'s invariant is violated, 
we perform a sequence of merges of adjacent runs in
our stack determined by $\mathcal{P}$
until the invariant is restored.
We call this a \emph{merge sequence}.
After processing all runs in $A$, the main loop exits and we perform a
\emph{collapse} phase on the remaining runs in the stack, where we
repeatedly merge the top two runs in the stack until
only one run remains, our final sorted list.
We use the convention that runs closer to the top of the stack
are labeled with lower indices%
, such that
the run at the top of the stack takes label $R_1$.
We use $r_i$ to denote the length of run $R_i$,
i.e., $r_i=|R_i|$.
We present pseudocode of a generic stack-based mergesort in
Algorithm~\ref{alg:generic-stack}.
See Figure~\ref{fig:stack}.

\begin{algorithm}[hbt]
	\caption{A Generic Stack-Based Mergesort Algorithm with Merge Policy, $\mathcal P$}\label{alg:generic-stack}
	\begin{algorithmic}[1]
		\State{\textbf{Input:} array $A$ of size $n$. Merge policy $\mathcal P$.}
        \State{\textbf{Output:} sorted array $A'$.}

        \State $\mathcal R \gets$ run decomposition of $A$
        \State $S \gets \emptyset$ \Comment{the stack of runs}
        \While{$\mathcal R$ is not empty} \Comment{main loop}
			\State Push next run from $\mathcal R$ onto $S$
            \While{$S$ violates $\mathcal P$'s invariant}
            \label{line:merge-policy-loop} \Comment{algorithm-specific merge policy}
                \State Merge a pair of adjacent runs in $S$ 
                   according to $\mathcal P$
            \EndWhile
        \EndWhile
        \While{$|S| > 1$} \Comment{collapse phase}
            \State Pop $R_1$ and $R_2$ from $S$
            and push \textsc{Merge}$(R_1,R_2)$ onto $S$
        \EndWhile
        \State \Return $\text{pop}(S)$
	\end{algorithmic}
\end{algorithm}

\begin{figure}[b]
  \centering
  \vspace*{-\medskipamount}
  \includegraphics[width=0.8\linewidth]{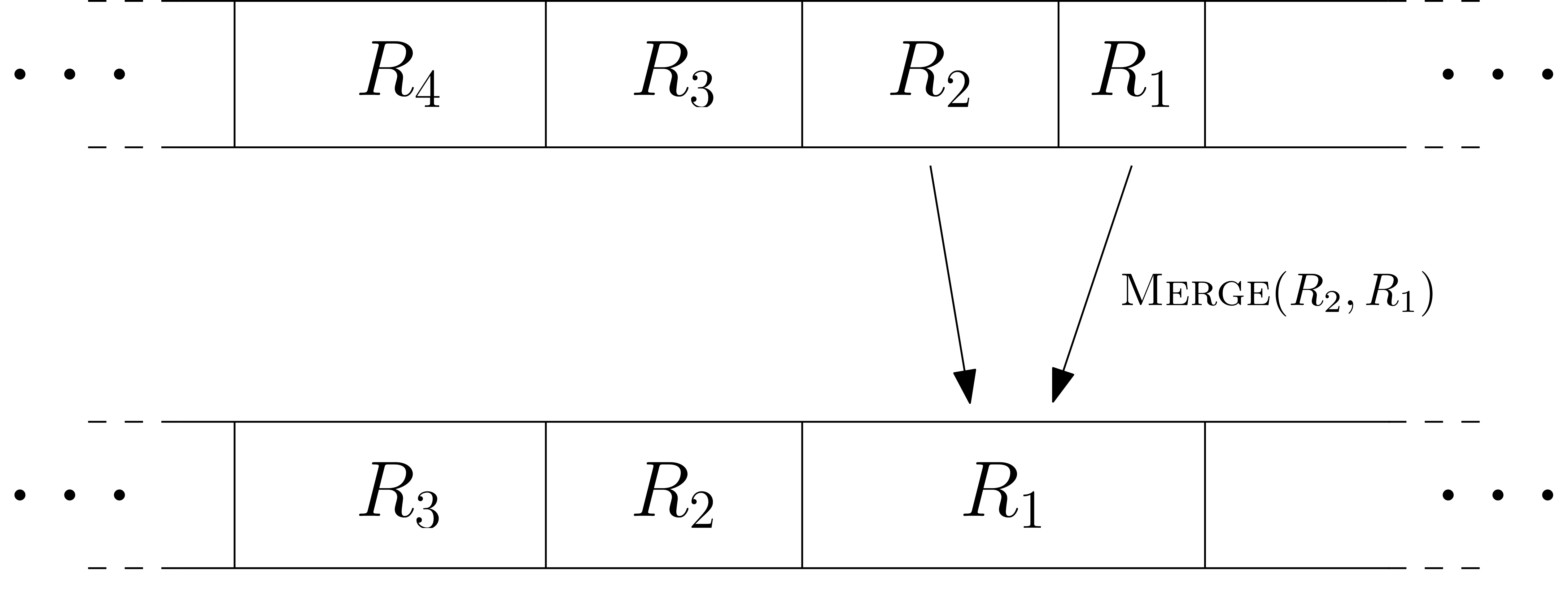}
  \caption{A merge step according to a merge policy, $\mathcal{P}$. 
  The merged run takes the label of the higher run. Merging decreases
  the labels of all runs below the two merged runs.}
  \label{fig:stack}
\end{figure}

\vspace*{-\medskipamount}

For ease of expression, we may still use the term ``stack'' 
when referring to a given stack-based mergesort, even though
our algorithms do not explicitly maintain all of $S$.
We review the merge policies of specific stack-based 
mergesorts as needed throughout the paper.

\subsection{Shallow Mergesorts}

Interestingly, most instance-optimal mergesorts have merge policies that
only consider a constant number, $k$, of runs from the top of the stack. 
Buss and Knop~\cite{buss19} made a similar observation, 
which led them to formulate the
notion of \emph{$k$-aware} mergesorts. 
A stack-based mergesort is 
\emph{$k$-aware} if its merge conditions 
are based on the run lengths of the top $k$ runs in the stack 
and it only merges runs 
in the top $k$ entries of the stack. For our purposes, 
we can also allow merge conditions to operate on other metadata or properties 
of runs that can be recovered in constant time,
such as their position in the array, as well as other global data,
provided all of this fits in $O(1)$ memory cells. 
We call this relaxed definition \emph{almost-$k$-aware}. 
Any $k$-aware algorithm is trivially almost-$k$-aware.%

These definitions suggests a very simple way to make 
almost-$k$-aware algorithms in-place:
replace our stack $S$ with a ``shallow stack'' that only 
maintains the top $k$ runs.
Of course, as runs are merged, other runs whose lengths we have forgotten
enter the top $k$ entries of $S$. We must 
determine the length of these forgotten runs to
test whether the merge policy's invariant
is still violated.
In this paper, we propose two simple algorithms that efficiently 
do so, the walk-back algorithm and
the jump-back algorithm.
They are described respectively in \Cref{sec:walkable} and
\Cref{sec:encoding}. 
If we can efficiently implement some
stack-based mergesort, $\mathcal S$, in-place
using a shallow stack, then we call $\mathcal S$ a 
\emph{shallow} mergesort.
We summarize our results and related properties of known 
algorithms in Table~\ref{tab:summary}.

\begin{table*}[b]
  \begin{center}
  \begin{tabular}{|c|c|c|c|c|}
    \hline
    {\bf Algorithm} & {\bf Shannon-Optimal} & {\bf Walkable} & {\bf Jumpable} & {\bf Shallow}\\
    \hline
    PowerSort~\cite{munro18} & $\checkmark$ & $\checkmark$ & $\checkmark$ & $\checkmark$ \\
    \hline
    $c$-Adaptive ShiversSort~\cite{juge24} & $\checkmark$ & $\checkmark$ & $\checkmark$ & $\checkmark$   \\
    \hline
    TimSort~\cite{tim,auger2019} & $\checkmark$ & \ding{55} & $\checkmark$ & $\checkmark$  \\
    \hline
    $\alpha$-StackSort~\cite{buss19} & \ding{55} & $\checkmark$ & $\checkmark$ & $\checkmark$ \\
    \hline
    $2$-MergeSort~\cite{buss19} & \ding{55} & $\checkmark$ & $\checkmark$ & $\checkmark$ \\
    \hline
    $\alpha$-MergeSort ($\varphi < \alpha < 2$) \cite{buss19,ghasemi2025} & $\checkmark$ & \ding{55} & $\checkmark$ & $\checkmark$\\
    \hline
    ShiversSort~\cite{buss19} & \ding{55} & $\checkmark$ & $\checkmark$ & $\checkmark$ \\
    \hline
  \end{tabular}
\end{center}
  \caption{A summary of our results applied to selected stack-based mergesorts. 
  }
  \label{tab:summary}
\end{table*}

\section{The Walk-Back Algorithm}\label{sec:walkable}
In this section, we define the walk-back algorithm and 
show that a variety of stack-based mergesorts are \emph{walkable},
i.e., applying the walk-back algorithm to them increases 
runtime by no more than a constant factor. 
The walk-back algorithm implements a shallow stack like so: whenever
we must determine the length of a run for a comparison, we start 
at the beginning of the run above it in the stack and walk backwards. 
See \Cref{fig:walking-backwards}. 

\begin{figure}[hbt]
    \centering
    \includegraphics[width=.8\linewidth]{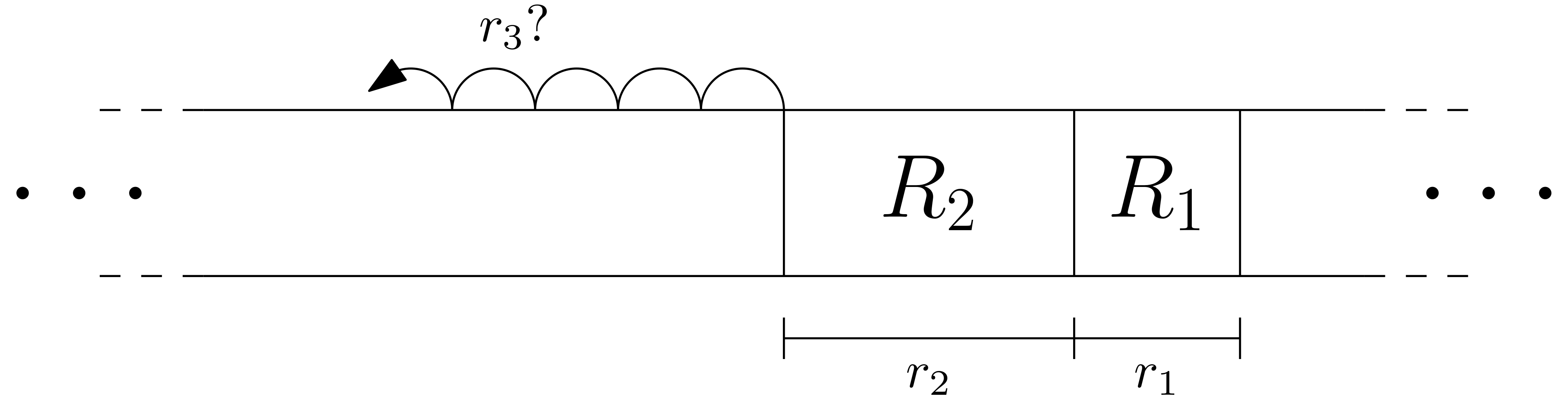}
    \caption{The walk-back algorithm applied to an almost-3-aware
    mergesort. 
    We currently know the lengths of the runs labeled $R_1$ and
    $R_2$. 
    There may be $O(\log n)$ runs to the left of $R_2$ in the array (our
    ``stack'' of runs), but we do not maintain any information about them.
    Whenever we need to know 
    the length of $R_3$, we start walking backwards from the start of $R_2$. Surprisingly, this strategy only increases runtime by a constant factor in several stack-based mergesorts.}
    \label{fig:walking-backwards}
\end{figure}

\subsection{The Walk-Back Procedure and Stopping Condition}
We keep walking until either (1) we determine the length of the run in
question or (2) we walk far enough to confirm 
whether the given comparison is false. 
We call (2) our \emph{stopping condition}. It varies by algorithm.
If we meet our stopping condition, we forget any information about how far we walked and retry from the
beginning the next time we invoke the walk-back algorithm. 

In the remainder of this section, we prove that two instance-optimal algorithms, 
PowerSort and $c$-Adaptive ShiversSort, are walkable.
In combination with an in-place merge algorithm,
this produces the first in-place instance-optimal mergesort algorithm.

Let us denote runs from the initial run decomposition $\mathcal R$ as \emph{original runs},
and those obtained as the result of merges as \emph{intermediate runs}.
Let $m$ be the sum of the sizes of all intermediate runs; 
the sum of the sizes of all original runs is $n$ by definition. For any stack-based mergesort algorithm $\mathcal S$, 
its runtime $T_\mathcal S(n) \geq m + n$. 
This is because, for all intermediate runs $i$, 
we must have spent $|i|$ time to merge the two runs that became $i$. 
In addition, any sorting algorithm must spend $n$ time just to read the input. 
If we can show that the walk-back algorithm applied to $\mathcal S$ 
adds no more than $O(m + n)$ time over
any input sequence, 
then we have shown that $\mathcal S$ is walkable. 

\begin{observation}\label{obs:always-know}
    The length of the topmost run, $r_1$, is always known without need of
    walking back during the execution of any stack-based mergesort,
    $\mathcal{S}$.
\end{observation}



\begin{lemma}\label{lem:walkable-collapse}
The collapse phase of any stack-based mergesort $\mathcal S$ can be implemented in-place
using the walk-back algorithm with only a constant-factor increase in runtime.
\end{lemma}

The proof follows from the fact that the cost of walking back to recover the
last two runs is proportional to the cost to merge them.
Because of the above lemma,
we only consider the main loop of each stack-based
mergesort for the remainder of this work. 
\subsection{PowerSort is Walkable}\label{sec:powersort-body}

In this section, we show that PowerSort, an instance-optimal stack-based
mergesort, is walkable
by seeing how the walk-back cost is `paid for' by the cost already spent on
merges.
In particular, we will analyze the walk-back cost of successful merges,
determine a stopping condition, and use this stopping condition to compute the
walk-back cost of failed merges.

Munro and Wild's PowerSort~\cite{munro18} assigns every run a
\emph{power} value, computed jointly with the run who sits above it
on the stack.
These power values are used to determine whether the second and third
runs on the stack should be merged.
%
To understand how runs are assigned powers, imagine superimposing a perfect
binary tree on top of the input array, where each index in the array is
associated with a binary tree node.\footnote{We assume for simplicity that the input
is of size $n = 2^a - 1$ for some $a$, although the result holds for
general $n$.}
Each index in the array is associated with a power value, where the power of an
index is the depth of the corresponding tree node.
The center index has a power 0, the first and third quartiles have a power 1,
and so on down to $\lfloor \log{n} \rfloor$.
Each run is associated with the smallest power value of any index between its
midpoint and the midpoint of the next run, which we refer to as the runs
\emph{power interval} $I$.
See Figure~\ref{fig:power} for an example of how power values appear and how to compute them.
Note that \cite{munro18} define power slightly differently, and in general
those power values are one greater than ours.
The resulting algorithm is the same, however.
We provide a simplified sketch of
PowerSort's merge policy in \Cref{alg:powersort}.

\begin{algorithm}[b]
	\caption{PowerSort's Merge Policy (see \Cref{line:merge-policy-loop} of \Cref{alg:generic-stack})}\label{alg:powersort}
	\begin{algorithmic}[1]
		\If{$|S| > 1$}
			\State $p_2 \gets \Call{NodePower}{R_2, R_1}$
            \Comment{$p_3$ already known}
			\While{$|S| > 2$ and $p_3 > p_2$}
			\State $R_2 \gets \Call{Merge}{R_3, R_2}$ 
            \State $p_3 \gets \Call{NodePower}{R_3, R_2}$
			\EndWhile
            \State $p_3 \gets p_2$
		\EndIf
	\end{algorithmic}
\end{algorithm}

\begin{figure}
    \centering
    \includegraphics[width=0.7\linewidth]{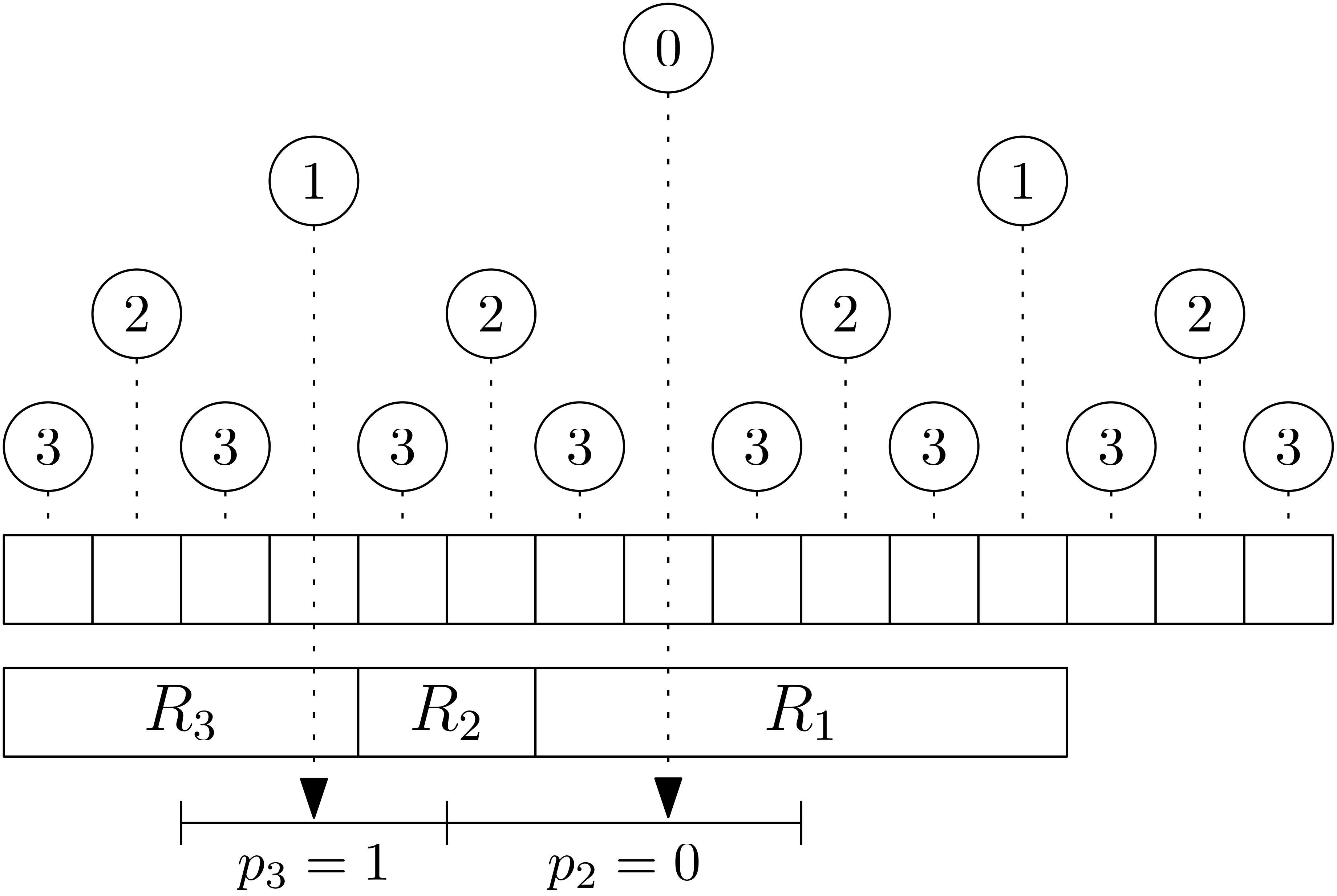}
    \caption{
		Each index in the array is associated with a \emph{power} value, where
		the power of an index is the depth of the corresponding tree node
		(depicted above it).
		Each run is associated with the smallest power value of any index
		between its midpoint and the midpoint of the next run, i.e. its \emph{power interval} $I$.
		The power intervals $I_3$ and $I_2$ define powers $p_3$ and $p_2$. They are depicted under their respective runs.
		$R_1$, being the top run, does not yet have an assigned power.
	}
    \label{fig:power}
\end{figure}

One barrier to applying the walk-back algorithm is that Munro and Wild's original algorithm does not recompute power values when merges occur. 
Since as many as $\Omega(\log n)$ power values 
may be stored at a time in the original PowerSort algorithm, 
maintaining a shallow stack necessitates recomputing power values as 
runs rise to the top of the stack.
It could be that such recomputed power values do not equal their original values,
causing the behavior of a shallow PowerSort to differ 
from standard PowerSort.
Fortunately, we show that this is not the case and power values computed on-the-fly are faithful to their original values.
\begin{lemma} \label{lem:powers-unchanged}
	The merge step of PowerSort (line 4 of \Cref{alg:powersort}) does not change
	the power of the second run, nor the power of any existing run.
\end{lemma}
\begin{proof}
	See \Cref{fig:dont-change}.
    A successful merge combines runs $R_2$ and $R_3$ where $p_3 > p_2$.
    Let $I_2'$ denote $R_2$'s power interval after the merge.
    Since merging only extends $I_2$ leftward into what was $I_3$ and all indices in $I_3$ have power $\geq p_3 > p_2$, 
    the power of $R_2$ remains $p_2$.
    Similarly, $R_4$'s interval extends rightward into $I_3$, gaining only indices with power $\geq p_3 \geq p_4$,
    so its power is unchanged.
    Deeper runs are unaffected.
    \qed
\end{proof}

\begin{figure}[t]
    \centering
    \includegraphics[width=0.7\linewidth]{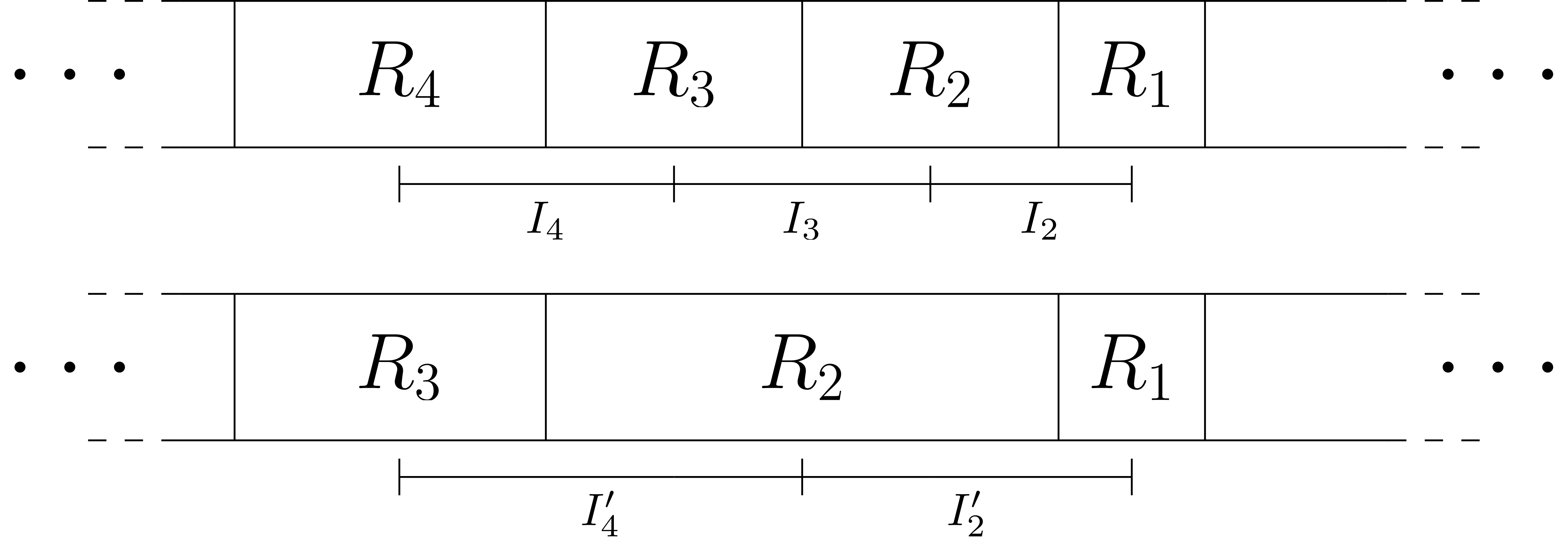}
    \caption{Merging
	runs $R_2$ and $R_3$
	causes $I_2$ and $I_4$ to extend into what used to be $I_3$. 
    Nevertheless, when $R_2$ and $R_3$ are merged, $p_3 > p_2$ and $p_3 \geq p_4$.  Therefore, the powers of $I'_4$ and $I'_2$ do not change.}
    \label{fig:dont-change}
\end{figure}

From Lemma~\ref{lem:powers-unchanged}, it follows that
PowerSort is almost-$3$-aware since we can correctly recompute all power values
with only knowledge of their \textit{current} power intervals. 
Given the lengths of the top-3 runs and the position of the topmost run, 
we can compute power intervals in constant time. 
By \Cref{obs:always-know}, we always know $r_1$. It costs at most $r_2 + r_3$ time to recover the lengths of the runs labeled $R_2$ and $R_3$ by walking back. 
If the merge condition is successful, 
then the merge cost of $R_2$ and $R_3$ equals our walk-back cost, 
so the total walk-back cost in this case is at most $m$. 
Thus, we have shown that successful merges ``pay'' for their own walk-back cost.
It remains to bound the walk-back cost of failed merges. 
To do so, we next develop a stopping condition for walk-back PowerSort.
We begin with the following observations about powers.

\begin{observation}
	Each power $p$ is evenly spaced in the sequence.
	Two indices with the same power $p > 0$ are separated by
	at least $n/2^p$ elements.\footnote{
    The distance here and in other observations may be off by
	at most one depending on the exact input size $n$.}
\end{observation}

\begin{observation} \label{obs:power-interval}
    The distance between any index with power 
    $p$ and the closest
    index with power less than $p$ is $n/2^{p + 1}$ 
    on either side.
	Consequently, any range of $n/2^{p + 1}$ indices contains either a power $p$
	or a power less than $p$.
\end{observation}

Both follow directly from our description of powers implicitly held at indices.

\begin{observation}\label{obs:run-size}
Any run with size at least $n/2^p$ has a power of at most $p$.
\end{observation}

This observation follows from the fact that the right half of this run, which
is at least size $n/2^{p - 1}$, is part of its power interval, and from the second
part of Observation~\ref{obs:power-interval}.
%
%
Rearranging, we obtain the following.

\begin{observation}
A run of size $r$ must have power at most $p \leq \log{\frac n r}$.
\end{observation}

It would be convenient if we could show that long runs always have small powers and short runs always have large powers. However, short runs can
have arbitrarily small powers if
they appear at 
indices of small power.
Such  ``lucky'' short runs may incur a disproportionate
walk-back cost for failed merge comparisons since
their eventual merge cost could fail to be proportional
to this walk-back cost. 
Despite this, we 
show that the overall walk-back cost over the entire main
loop is small.

\begin{observation}\label{obs:merge-distance}
	Consider a merge of runs $R_3$ and $R_2$ in the merge step of the
	algorithm.
	The distance between the midpoints of $R_3$ and $R_1$ is at least
	$n/2^{p_3 + 1}$, where $p_3$ was the power of $R_3$ before the merge.
\end{observation}

\begin{proof}
    The range of indices 
    between the midpoints of $R_3$ and $R_1$ comprises the power intervals $I_3$ and $I_2$.
	Since the merge occurred, $p_3 > p_2$. 
    From the first part of
	Observation~\ref{obs:power-interval}, the distance between the indices of
	these two powers is at least $n/2^{p_3 + 1}$.
	The proof results from the fact that the nodes containing $p_2$ and $p_3$
	are in $I_2$ and $I_3$ respectively.
	\qed
\end{proof}

\begin{lemma}\label{lem:powersort-stopping}
If $r_3 \geq n/2^{p_2}$, then $p_3 \leq p_2$.
\end{lemma}
\begin{proof}
Follows directly from \Cref{obs:run-size}.
\qed
\end{proof}

\begin{lemma}\label{lem:powersort-r2}
During PowerSort's main loop, so long as the stack has at least two entries, we always know $r_2$.
\end{lemma}
\begin{proof}
Any new $R_2$ was either previously labeled
$R_1$, whose length is known (\Cref{obs:always-know}), or 
is the result of a
merge.
\qed
\end{proof}

From \Cref{obs:always-know} and \Cref{lem:powersort-r2}, we always know $r_1$ and $r_2$ 
when testing PowerSort's merge condition. 
Thus, \Cref{lem:powersort-stopping} defines our stopping condition and shows that failed merges incur a walk-back cost of at most $n/2^{p_2}$.

\begin{lemma}\label{lem:powersort-walkback-cost}
The total walk-back cost of failed merges in PowerSort is $O(m+n)$. 
\end{lemma}
\begin{proof}
Consider two runs, $q$ and $s$, such that $s$ appears directly above $q$ in the
stack.
Their respective powers are $p_q$ and $p_s$. Say $q$ currently has label $R_3$,
$s$ has label $R_2$,
and we are about to test the merge condition between them.
Say we find that $|q| > n/2^{p_s}$;
by \Cref{lem:powersort-stopping}, the merge condition fails, we exit the walk-back algorithm and terminate the merge sequence. What can we charge this walk-back cost of $n/2^{p_s}$ to?

We consider two cases. In case 1, $q$ is never part of a failed merge again.
Then we can charge this cost to $q$ itself since $|q| > n/2^{p_s}$. Over all
runs, $q$, this adds at most $(m+n)$ to the walk-back cost as $q$ may be
original or intermediate. 

In case 2, $q$ will be part of another failed merge. This means that $q$ will have to take the label $R_3$ again, which implies that $s$ gets merged during the main loop. When $s$ is merged, $q$ has label $R_4$ and $s$ has label $R_3$. By \Cref{obs:merge-distance}, 
the distance between the midpoints of $R_3$ and $R_1$ is at least $n/2^{p_s+1}$. 
This implies that $r_1/2 + r_2 + r_3/2 \geq n/2^{p_s+1}$, so $r_1 + 2r_2 + r_3 \geq n/2^{p_s}$. For simplicity, we upper bound this as $2(r_1 + r_2 + r_3)$. 

We can charge our cost of $n/2^{p_s}$ to twice the sizes of the current runs with labels $R_1$, $R_2$, and $R_3$. 
All three runs are charged in this way once. For $R_2$ and $R_3$, this is because they are immediately merged and so destroyed. 
For $R_1$, recall that we assume that $q$ fails its merge condition, so the merge sequence is subsequently terminated. 
When a merge sequence finishes, we push a new run onto the stack, shifting the run labeled $R_1$ to $R_2$. There is no way for this run to ascend back to $R_1$, so it is charged as $R_1$ once.

As a result, this case charges each run seen during the main loop twice its length at most three times (i.e., when labeled $R_1$, $R_2$, and $R_3$). 
This contributes an additional $O(m+n)$ to the total walk-back cost of failed
merges.
\qed
\end{proof}

We have shown that, over all merges, the walk-back cost is bounded by $O(m + n)$. We conclude with the following.
\begin{theorem}\label{thm:powersort}
	PowerSort is walkable.
\end{theorem}

\begin{corollary}
We can implement PowerSort with the walk-back algorithm to produce an in-place, stable sorting algorithm with running time $\Theta(n(1 + \mathcal H(A)))$.
\end{corollary}
\begin{proof}
Follows from properties of PowerSort, stable in-place merge algorithms, and the
walk-back algorithm.
\qed
\end{proof}

\section{The Jump-Back Algorithm}
\label{sec:encoding}
In this section, we describe our jump-back algorithm,
which is a general method for implementing all the known 
stack-based natural mergesort algorithms to be in-place without a stack, albeit
while sacrificing stability.
The only assumption
required is that a stack-based
mergesort algorithm, $\mathcal{S}$, 
be almost-$k$-aware.
Importantly, if the running time of $\mathcal{S}$ is
$O(n(1 + \mathcal H(A)))$, then the running time of our in-place
jump-back implementation of $\mathcal{S}$ will also
run in time
$O(n(1 + \mathcal H(A)))$.

Suppose we are given an input array, $A$, of size $n$,
and let $\lambda=\lceil\log n\rceil + 1$ ($\lambda$ is a parameter we that
we use throughout our jump-back algorithm).
We begin by moving the elements in
\emph{short} runs, those of size at most $3\lambda$, to the end of $A$.
We then sort them via a standard in-place mergesort algorithm; e.g., 
see~\cite{edelkamp2020quickxsort,edelkamp2019worst,huang1992fast,DALKILIC20111049,CHEN2003191,ellis}.
Our method for performing this move in-place can be done
in linear time
using a partitioning algorithm described in \Cref{sec:compression}.
%
%
See Figure~\ref{fig:compression}.

\begin{figure}[hbt]
    \centering
    \includegraphics[width=0.8\linewidth]{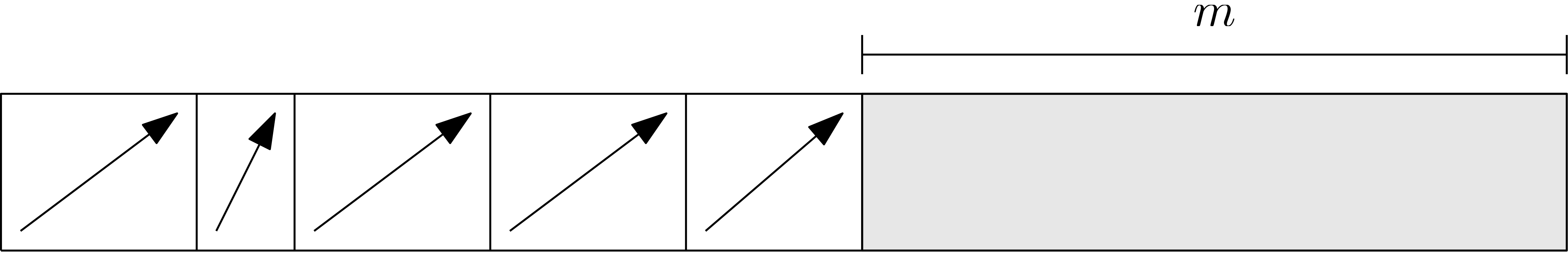}
    \caption{The result of performing partitioning to move all short runs to the
    end of the array. All runs to the left of the gray block are of size $>
    3\lambda$.
    The elements in the gray block belonged to runs of size $\leq 3\lambda$. Their ordering may not be preserved. }
    \label{fig:compression}
\end{figure}

We then apply the given stack-based natural mergesort algorithm $\mathcal{S}$
to the remaining, \emph{long}, runs, with one change---rather than using a complete 
stack to store
run sizes, as in the original way $\mathcal{S}$ was defined,
we use a shallow stack of size $k$ and 
we use bit-encoding to encode the size of each run {\it in the run itself}.

\begin{lemma} \label{lem:encoding}
    The size of any \emph{long} run can be encoded in the run itself using a
    reversible in-place encoding scheme that can be encoded, decoded, and
    reversed in $O(\lambda)=O(\log n)$ time.
\end{lemma}


This lemma is proved by the introduction of two new encoding schemes which we
call \emph{pivot-encoding}, which works when there are a sufficient number of
distinct elements in the run, and \emph{marker-encoding}, which works in general.
Under either scheme, we encode the size of a run within the last
$\lambda + 1$ elements of the run.
We discuss these encoding schemes in
\Cref{sec:encoding-app}.
%
%
See Figure~\ref{fig:walking-encodings}.

\begin{figure}[b]
    \centering
    \includegraphics[width=0.7\linewidth]{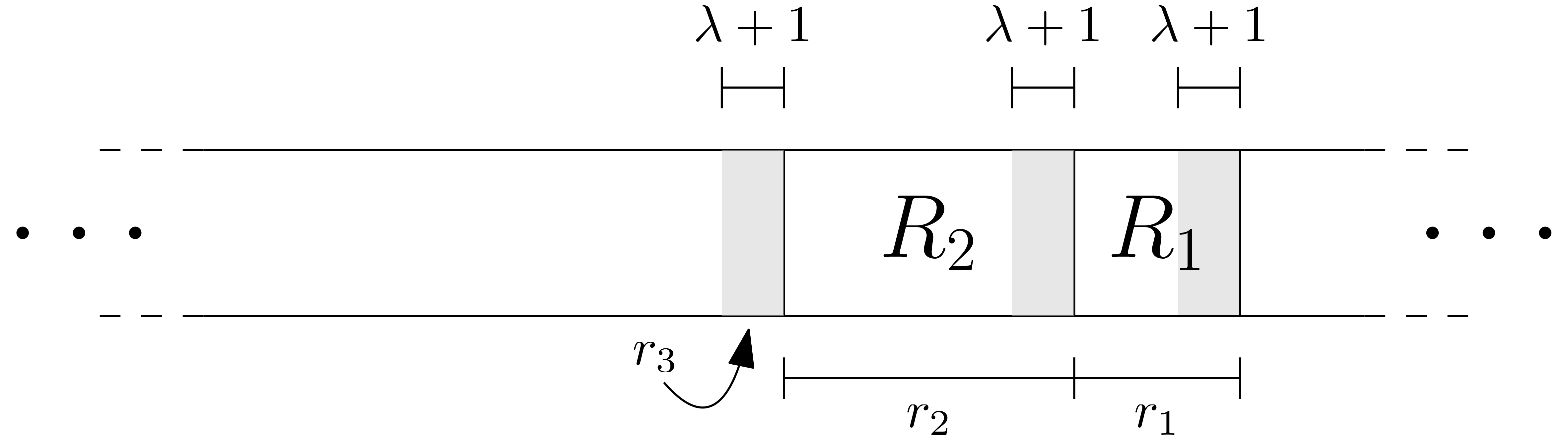}
    \caption{We encode run lengths
    within the last $\lambda + 1$ elements of each run, indicated with shading,
    allowing us to 
    compute the run's length in $O(\log n)$ time. This eliminates the need for walking back via the walk-back algorithm albeit at the loss of stability.}
\label{fig:walking-encodings}
\end{figure}



Recall that we have already moved all short runs to the end of the array.
Using \Cref{lem:encoding}, we can implement the almost-$k$-aware algorithm,
$\mathcal{S}$, in-place on the remaining long runs
using our jump-back algorithm with only two
changes.
%
\begin{enumerate}
    \item Whenever a new run, $q$, is added and the bottommost run, $s$, in the shallow
    stack no longer fits, we \textit{encode} $s$'s size using the algorithm of
    \Cref{lem:encoding}.
    This step takes $O(\lambda)$ time, and can be charged to the run, $q$,
    that was just added to the stack, since its size must have been at least
    $3\lambda + 1$.

    \item Symmetrically, whenever a run is merged and the shallow stack is no longer
    full, we add the run deeper in the array to the shallow stack by
    \textit{decoding} its size and by \textit{reversing} the encoding
    such that the original run is restored.
    This step also takes $O(\lambda)$ time, and can be charged to the merge
    operation.
\end{enumerate}


Using these two changes, we can maintain all the information of our shallow
stack efficiently, without asymptotically increasing the running time of our
algorithm, $\mathcal{S}$.
All that is left is to
sort
all the remaining
\textit{short} runs.

All short runs, of size at most $3\lambda$, were initially moved
to the end of the array.
Let us denote this portion of the array as $A'$.
By the definition of the Shannon entropy of an array,
%
$n\mathcal H(A) = \sum_{i=1}^{\rho(A)} r_i\log(n/r_i)$, and we see that
\[
\sum_{i=1}^{\rho(A')}r_i\log(n/r_i) \geq
\sum_{i=1}^{\rho(A')}r_i\log(n/3\lambda).
\]
But $\lambda=\lceil\log n\rceil + 1$; hence,
if we let $m = \sum_{i=1}^{\rho(A')}r_i$, then the sum of their entropy is
$\Omega(m\log n)$. 
As a result, the cost of sorting all runs of
size at most $3\lambda$ in time $O(m\log m)$ is ``paid'' for by their
contribution to the entropy of the original array. We spend an
additional $O(n)$ time moving all such runs to the end of the array.
(see \Cref{app:in-place}).
%
As noted above, the time to sort the long runs in $A$ is
asymptotically the same as $\mathcal{S}$.
Thus, we have the following:


\begin{theorem}
Given an array, $A$, of size $n$, and a stack-based natural mergesort
algorithm, $\mathcal{S}$, which is almost-$k$-aware for constant $k\ge2$,
we can implement
$A$ in-place in $O(n(1+\mathcal{H}(A))+T_{\mathcal{S}}(n))$
time, where $T_{\mathcal{S}}(n)$ is the running time of $\mathcal{S}$.
\end{theorem}

\section{Our Bit Encoding Methods}\label{sec:encoding-app}



In this section, we describe our methods for encoding run sizes in-place
in the runs themselves.
Unlike previous encoding methods, such as those of Kuszmaul and
Westover~\cite{kuszmaul2020place}, 
our methods make no assumptions about the elements and their 
representation themselves, except that they are
comparable and copyable.
Recall that we define $\lambda = \lceil \log n\rceil + 1$. We do so because the size of each run can be represented using $\lceil \log n\rceil$ bits, 
along with one extra bit to specify the encoding method used.
More specifically, let $b_1b_2\ldots b_{\lceil \log n\rceil}$ be a binary
encoding of the size, $r$, of a \textit{long} run, $R$, and let $b_\lambda$
indicate the encoding method used, i.e., $b_\lambda=0$ for the
\emph{pivot-encoding} method and $b_\lambda=1$ for the
\emph{marker-encoding} method.
Both of our methods efficiently encode $\vec{b}$ in-place within the run, $R$.


\paragraph{Our Pivot-Encoding Method.}
Our first encoding method is a simple encoding method we
call
the \emph{pivot-encoding} method.
Let $R$ be a \textit{long} run, and 
let $F$ and $L$ be the first and last $\lambda$ elements in $R$,
respectively.
Also, let $p$ be the element of $R$ that is immediately before $L$,
which we call the \emph{pivot} 
(i.e., $p$ is the element $\lambda+1$ cells from the end of $R$).
We say that the run $R$ is \emph{pivotable} if the last element of $F$, $F_\ell$, is
strictly smaller than our pivot, $p$. See \Cref{fig:pivotable}.

\begin{figure}[t]
    \centering
    \includegraphics[width=0.8\linewidth]{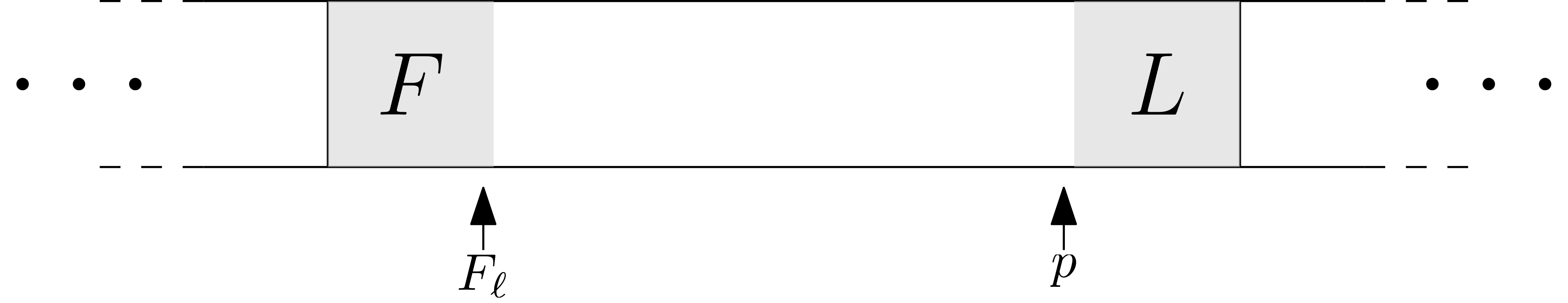}
    \caption{Here we depict the test to determine whether a run is \emph{pivotable}. We define $F$ and $L$ as the first and last $\lambda$ cells of the run. Our pivot $p$ is the element just prior to $L$ and $F_\ell$ is the last element of $F$. This run is pivotable if $F_\ell < p$.}
    \label{fig:pivotable}
\end{figure}

\begin{figure}[b]
    \centering
    \includegraphics[width=0.8\linewidth]{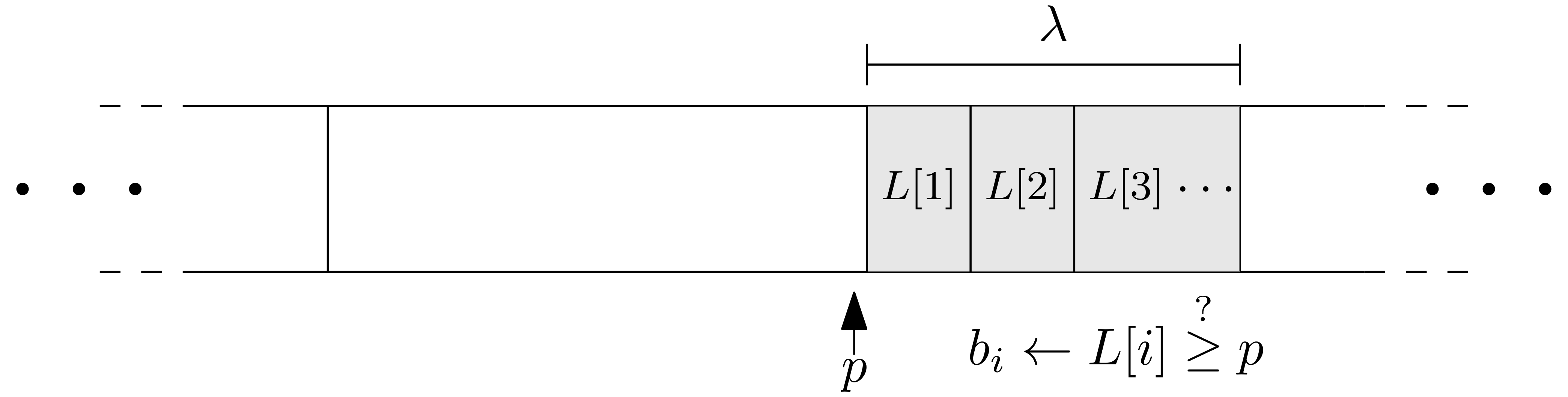}
    \caption{The decoding process, which applies to both our pivot-encoding and marker-encoding schemes. The last $\lambda$ elements encode our bit string $\vec b$. Element $p$, positioned $\lambda + 1$ elements from the end, determines the bit corresponding to each $L[i]$. When decoding, if $L[i] \geq p$, then $b_i$ is 1. Otherwise, it is 0.}
    \label{fig:encoding-retrieval}
\end{figure}

Suppose then that $R$ is pivotable, which implies that each element in $F$ is
strictly smaller than every element in $L$ and $p$.
Encoding $\vec b$
in our pivot-encoding scheme is easy:
We consider each element, $L[i]$, of $L$ as corresponding to the
bit, $b_i$, and, if $b_i=0$, then we swap $L[i]$ and $F[i]$.
Thus,
if $L[i]<p$, we know
that $b_i=0$, and if $L[i]\ge p$, then we know that $b_i=1$. See
\Cref{fig:encoding-retrieval}.
Encoding, reading, and reversing such an encoding can be done in $O(\log n)$
time.


\paragraph{Our Marker-Encoding Method.}
For a long run, $R$, not to be pivotable, i.e., $F_\ell \geq p$, it must be true that
not only $F_\ell = p$, but also that the entirety of the elements between $F$
and $L$, denoted $M$, are equal to $p$.
%
In this case, our pivot-encoding method does not work, and we use a different
encoding method, which we
call
the \emph{marker-encoding} method.

For our marker-encoding method, we designate the first two distinct elements
encountered as \emph{markers}, denoting them $m_1$ and $m_2$, where $m_1 < m_2$,
and store them using $O(1)$ memory.
These markers are in fact derived from the first run itself, since each run must
contain at least two distinct elements.

Since $R$ is not pivotable, we know that $F_\ell = p$, and thus all elements in $M$ are equal to $p$.
First, we store a copy of $L$ in the first $\lambda$ cells of $M$.
Next, we encode $\vec b$ in $L$ using only our two markers, letting $L[i]$
correspond to the bit $b_i$, and, if $b_i=0$, then we set
$L[i] \leftarrow m_1$ and if $b_i=1$, then we set
$L[i] \leftarrow m_2$.
Finally, we set the pivot value $p$ to $m_2$.
Reading this encoding is identical to how the pivot-encoding method is read,
since it is still true that if $L[i] < p$, then $b_i=0$ and if $L[i] \ge p$,
then $b_i=1$.
After the encoding is read and we must reverse the encoding, we simply set $L$
back to its original values using its copy in $M$, and then we clean up $M$
by copying $F_\ell$ into its first $\lambda$ elements and into the original pivot location, $p$.
See \Cref{fig:marker-encoding}.

\begin{figure}[t]
    \centering
    \includegraphics[width=0.8\linewidth]{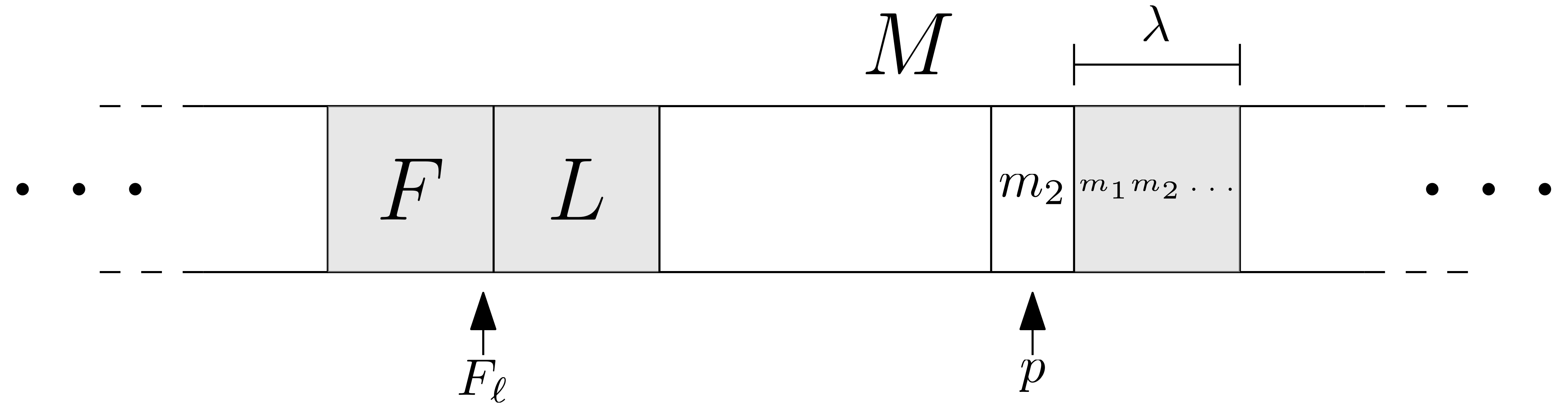}
    \caption{Transforming the run to apply marker-encoding. 
    We copy the block $L$ to just after $F$, with the remaining entries belonging to $M$. As in pivot-encoding, we use $p$ and the last $\lambda$ entries of the array to encode $\vec b$. However, this time it is done via our markers $m_1$ and $m_2$, rather than the elements of $F$ and $L$. We replace $p$'s element with $m_2$ to make the decoding process identical to that of pivot-encoding.}
    \label{fig:marker-encoding}
\end{figure}

\section{Conclusion}
We have shown how to implement stack-based natural mergesort algorithms to be strictly in-place and retain their
original asymptotic running times, including showing that several such algorithms, but not all,
can be implemented to be strictly in-place using a simple walk-back algorithm.
Our bit-encoding methods in our jump-back are immutable, in that they do not require modifying any bits of any elements,
unlike previous encoding methods.

\bibliographystyle{splncs04}
\bibliography{references}

\clearpage
\appendix
\section{$c$-Adaptive ShiversSort is Walkable}\label{sec:shivers-appendix}

In this section, we show that the recent instance-optimal stack-based mergesort, 
$c$-Adaptive ShiversSort~\cite{juge24}, is walkable. 
See~\Cref{alg:c-shivers} for its merge policy. 
The author defines $c$ as a tunable parameter which allows
for a tighter runtime analysis. 
Such details are not relevant to our work, however.

\begin{algorithm}[htb]
	\caption{$c$-Adaptive ShiversSort's Merge Policy (see \Cref{line:merge-policy-loop} of \Cref{alg:generic-stack})}\label{alg:c-shivers}
	\begin{algorithmic}[1]
        \State $\ell_1 \gets \lfloor\log(r_1 / c)\rfloor$ (resp. $\ell_2, \ell_3$)
        \While{$|S| \geq 3$ and $\ell_3 \leq \max\{\ell_2, \ell_1\}$}
            \State $R_2 \gets \Call{Merge}{R_3, R_2}$ \Comment{removes $R_3$ from $S$}
        \EndWhile
    \end{algorithmic}
\end{algorithm}

We first show that all successful merges during the main loop ``pay for'' the walk-back cost required to perform them.
Assume that the condition on line 2 of \Cref{alg:c-shivers} is true. From \Cref{obs:always-know}, we already know $r_1$. 
The walkback cost to perform this operation is at most $r_2 + r_3$. This is exactly the cost to merge
the runs labeled $R_2$ and $R_3$. As a result, the walk-back cost over all successful merges is bounded by the total merge cost throughout the main loop, which is at most $m$. 
It remains to bound the walk-back cost when the condition fails to hold. We first 
determine the stopping condition for $c$-Adaptive ShiversSort.

\begin{figure}[hbt]
    \centering
    \includegraphics[width=.8\linewidth]{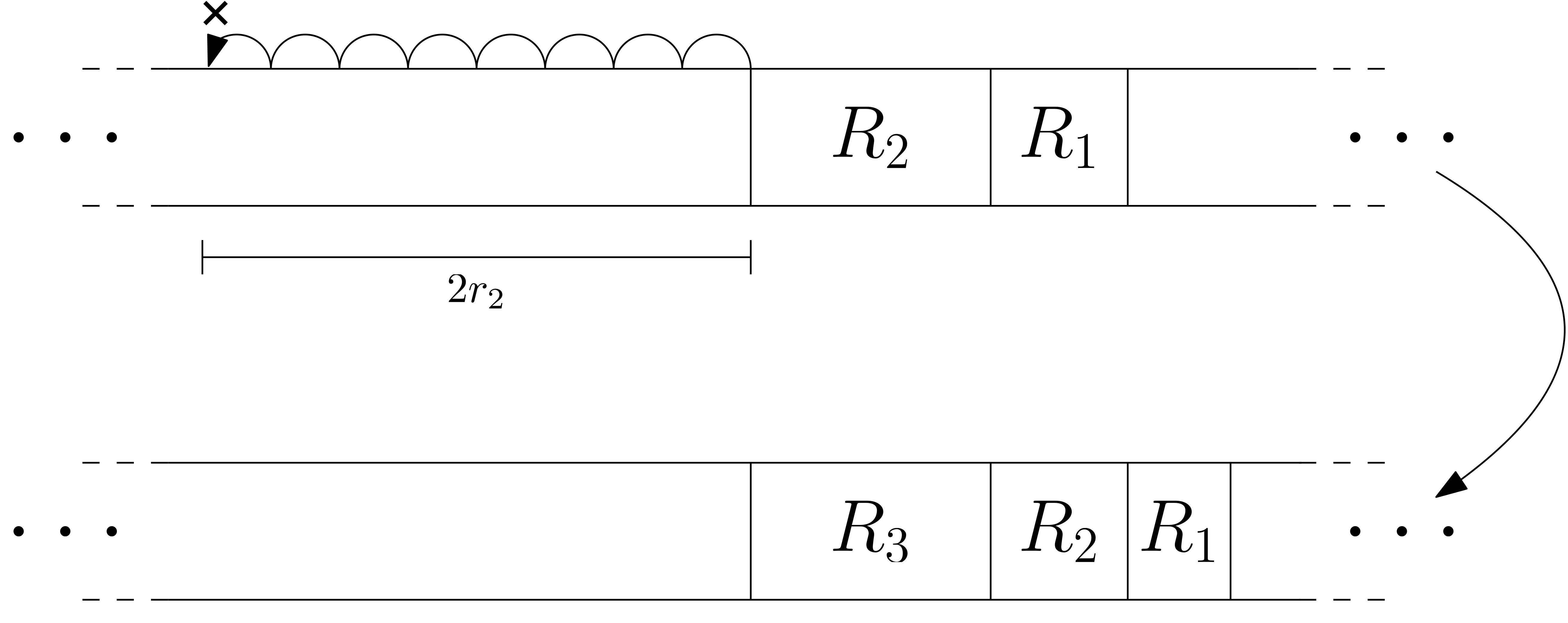}
    \caption{In this figure, we illustrate what happens when the merge condition
    fails in $c$-Adaptive ShiversSort. Assume w.l.o.g. that $\ell_2 =
    \max\{\ell_2, \ell_1\}$. We have walked back $2\times r_2$ steps yet have not
    found $r_3$. By \Cref{lem:stop-condition}, the merge condition must be false.
    Thus, we continue the main loop by pushing a new run onto the stack.
    }
    \label{fig:shiver-failure-case}
\end{figure}

\begin{lemma}\label{lem:stop-condition}
If $r_3 > 2\times \max\{r_1,r_2\}$, then $\ell_3 > \max\{\ell_1, \ell_2\}$.
\end{lemma}
\begin{proof}
Without loss of generality, let $r_2 = \max\{r_1,r_2\}$. Then $\ell_2 \geq \ell_1$.  
Since $r_3 > 2r_2$, the following holds true. 
$$\ell_3 \geq \lfloor \log(2r_2/c) \rfloor = \lfloor \log(r_2 / c) + 1 \rfloor = \lfloor \log(r_2 / c) \rfloor + 1 > \ell_2$$
\qed
\end{proof}

The above lemma gives us the stopping condition for the walk-back algorithm applied to $c$-Adaptive ShiversSort: if we walk back $2\times \max\{r_1,r_2\}$ steps and have not found $r_3$, then we know that the merge condition in line 2 of \Cref{alg:c-shivers} is false. See \Cref{fig:shiver-failure-case}. 

\begin{lemma}\label{lem:failed-cost}
The walk-back cost incurred by failed merges is $O(m + n)$. 
\end{lemma}
\begin{proof}
As described above, we spend $2\times \max\{r_1,r_2\}$ walking back when we fail a merge condition. Since no subsequent merge occurs, we do not immediately ``pay back'' the cost. For simplicity, we can upper bound this cost to at most $2\times (r_1 + r_2)$. Thus, every failed merge check incurs a walk-back cost of at most twice the length of the runs labeled $R_1$ and $R_2$. 

To complete the argument, we show that each run is labeled $R_1$ and $R_2$ for at most one merge sequence each throughout the main loop.
Consider any run $q$ with label $R_1$ or $R_2$ when a failed condition occurs. Assuming the run decomposition is not empty, we push a new run onto the stack and then repeat the main loop.
This relabels $q$ to $R_2$ or $R_3$, respectively. 
Since $R_1$ and $R_2$ are never merged in this algorithm, it is impossible to ascend from label $R_2$ to $R_1$. It is possible to change labels from $R_3$ to $R_2$ with a merge. Nevertheless, doing so destroys the constituent runs in the merge and replaces them with a new merged run. As a result, each run is labeled $R_1$ and $R_2$ for at most one merge sequence. 

The sum of all original and intermediate runs is $m + n$. Each run $q$ contributes a walk-back cost of $2|q|$ at most twice (once as $R_1$ and again as $R_2$). Then the total walk-back cost of failed merges is at most $4(m + n)$. 
\qed
\end{proof}

We have shown that the walk-back cost of successful merges and failed merges during the main loop of $c$-Adaptive ShiversSort is $O(m + n)$. In combination with \Cref{lem:walkable-collapse}, we conclude the following.
\begin{theorem}
$c$-Adaptive ShiversSort is walkable.
\end{theorem}

Upon further inspection, our proof relies on two key properties of $c$-Adaptive ShiversSort. First, we only ever merge $R_3$ and $R_2$ and never examine runs below $R_3$. Second, the walk-back cost of a failed merge condition is bounded by a constant multiple of our known runs. 
The first property allows us to easily show that the walk-back cost of successful merges is paid for by the subsequent merge. 
The second is a relaxed stopping condition that allows us to bound the walk-back cost of failed merges.
Indeed, with slight modification, the above proof applies to {\it any} almost-$k$-aware mergesort with the following two properties:
    \begin{enumerate}
        \item Only $R_k$ and $R_{k-1}$ are ever merged during the main loop.
        \item If $r_k > \alpha\sum_{i=1}^{k-1}r_i$ for some constant $\alpha$, no merge occurs. 
    \end{enumerate}

\begin{corollary}
    The original ShiversSort and $\alpha$-StackSort are walkable.
\end{corollary}

\begin{proof}
    Since both ShiversSort and $\alpha$-StackSort~\cite{buss19} are
    almost-$2$-aware, they trivially satisfy the first property.
    ShiversSort's merge condition, $2^{\lfloor \log{r_2} \rfloor} \leq
    r_1$, satisfies the second property with $\alpha=2$.
    $\alpha$-StackSort similarly does so with its own parameter $\alpha$. Thus,
    both stack-based mergesorts are walkable.
    \qed
\end{proof}
\section{TimSort is Not Walkable} \label{sec:counter}

It is tempting to try to extend our walk-back analyses to TimSort.
Unfortunately, in this section we show that TimSort is
not walkable by providing an input instance such that
TimSort implemented with the walk-back algorithm performs
asymptotically worse than standard, non-in-place TimSort.
This
motivates our jump-back algorithm, presented in \Cref{sec:encoding}, which works for virtually all
stack-based mergesorts.
In \Cref{sec:timsort-proof}, we use a similar argument to show that $\alpha$-MergeSort is also not walkable.

\begin{figure}
    \centering
    \includegraphics[width=0.9\linewidth]{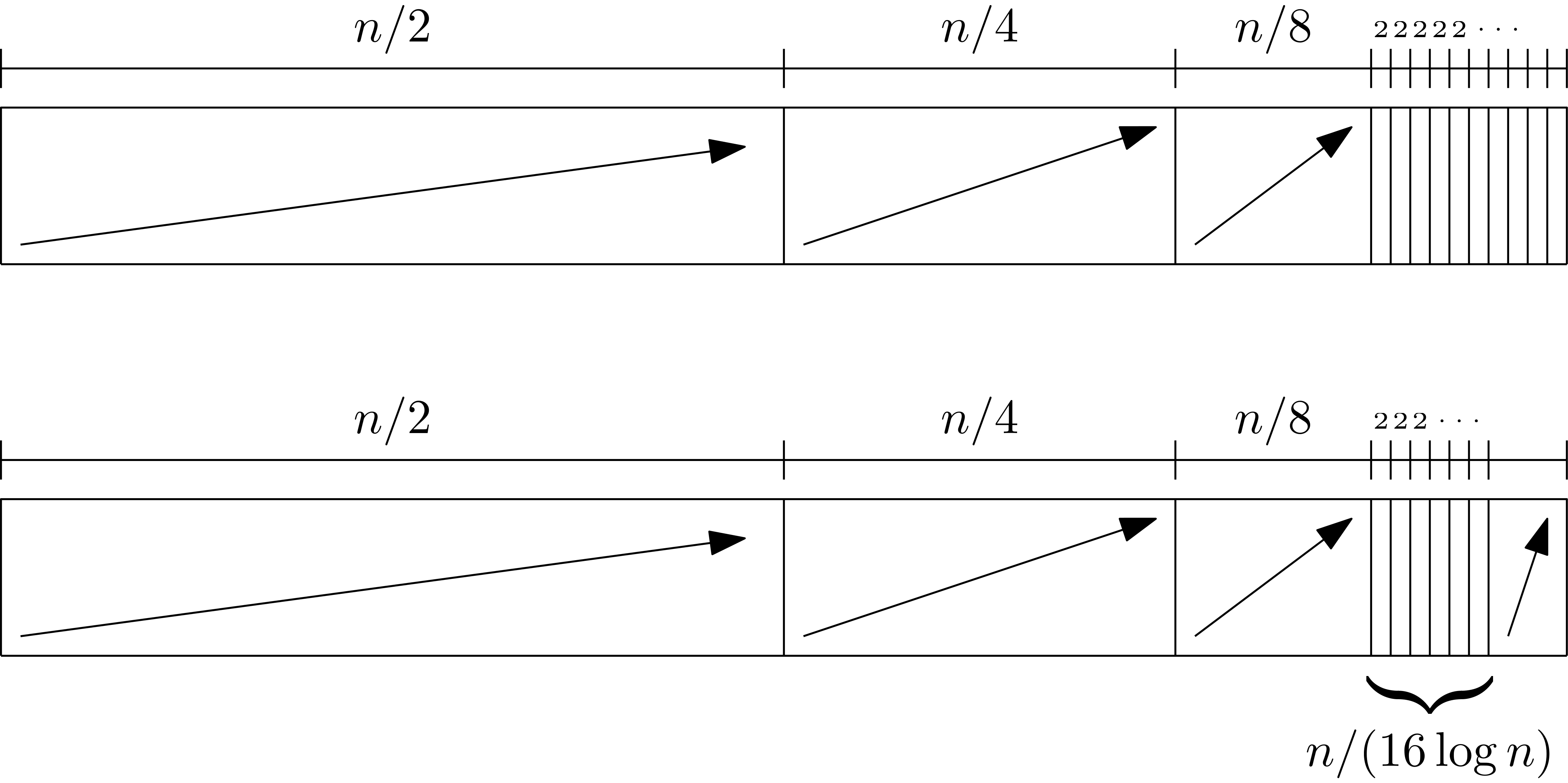}
    \caption{The worst-case input sequences we construct for our TimSort walk-back implementation. The top array represents the first version we analyze. The series of $n/16$ runs of size 2 cause $\Theta(\log n)$ {\it stages} to occur, blowing up walk-back cost to $\Omega(n\log n)$. The bottom array reduces the number of runs of size $2$ to $n/(16\log n)$. The walk-back cost remains $\Omega(n\log n)$ while the Shannon entropy of the array becomes constant. Therefore, TimSort with the walk-back algorithm will perform asymptotically worse than standard TimSort on the second sequence.}
    \label{fig:timsort-ctr}
\end{figure}

\begin{algorithm}
	\caption{TimSort's Merge Policy (see \Cref{line:merge-policy-loop} of \Cref{alg:generic-stack})}\label{alg:timsort}
	\begin{algorithmic}[1]
        \While{true}
            \If{$|S| > 3$ and $r_1 > r_3$} \Comment{merge condition \#1}
                \State $R_2 \gets \Call{Merge}{R_3, R_2}$
            \ElsIf{$|S| > 2$ and $r_1 \geq r_2$} \Comment{merge condition \#2}
                \State $R_1 \gets \Call{Merge}{R_2, R_1}$
            \ElsIf{$|S| > 3$ and $r_1 + r_2 \geq r_3$} \Comment{merge condition \#3}
                \State $R_1 \gets \Call{Merge}{R_2, R_1}$
            \ElsIf{$|S| > 4$ and $r_2 + r_3 \geq r_4$} \Comment{merge condition \#4}
                \State $R_1 \gets \Call{Merge}{R_2, R_1}$
            \Else
                \State \textbf{break}
            \EndIf
        \EndWhile
    \end{algorithmic}
\end{algorithm}

Our counterexample considers an array of length $n$ with the following preexisting runs. The first $n/2$ elements form a run, the next $n/4$ elements form a run, and the next $n/8$ elements form a run. The remaining elements form a sequence of $n/16$ runs, each of size 2. For simplicity, assume $n$ evenly divides into these runs. See \Cref{fig:timsort-ctr}. Our proof exploits the following key observations, which follow directly from TimSort's merge policy (see \Cref{alg:timsort}).

\begin{observation}\label{obs:merge-terminate}
If a merge sequence in TimSort produces a stack of runs whose lengths form a sequence of strictly decreasing powers of 2, then the merge sequence terminates.
\end{observation}
\begin{observation}\label{obs:timsort-validate}
We must know $r_2$ to validate that all merge conditions are false at the end of a merge sequence.
\end{observation}

\Cref{obs:merge-terminate} has two notable effects. First, the beginning three runs, those of lengths $n/2$, $n/4$, and $n/8$, are only ever merged at the very end of the main loop. Second, repeatedly adding runs of size 2 causes them to collapse into powers of 2. For example, adding the first four runs to the stack produces the series of run lengths $\{n/2,n/4,n/8,2\}$. This produces no merges. However, adding the next run produces a stack like so: $\{n/2,n/4,n/8,2,2\}$, which collapses into $\{n/2,n/4,n/8,4\}$ by merge condition \#2 of \Cref{alg:timsort}. We generalize this with the following definition.

\begin{definition}
Let \emph{stage $i$} of TimSort run on our input instance be defined as follows. Consider a set of runs on the stack with sizes $$\{n/2,n/4,n/8,2^{i-1},2^{i-2}, \ldots, 8, 4, 2\}.$$
Since their sizes are decreasing powers of 2, no merges occur by \Cref{obs:merge-terminate}. 
In the next iteration of the main loop, an additional run is added.
We say this is the beginning of stage $i$.
The run lengths are like so: 
$$\{n/2,n/4,n/8,2^{i-1},2^{i-2}, \ldots, 8, 4, 2,2\},$$ which
causes merge condition \#2 to be repeatedly applied. For $i < \log(n/8)$ and by \Cref{obs:merge-terminate}, the merge sequence terminates when the stack is of size four, with run lengths like so: 
$$\{n/2,n/4,n/8,2^{i}\}.$$ 
\end{definition}

Furthermore, let $k = O(1)$ be the size of our shallow stack. As defined by the walk-back algorithm, we only maintain the top $k$ elements of the stack. The rest are forgotten.

\begin{observation}\label{obs:stage-size}
At the beginning of stage $i$, there are $i+3$ runs in the stack. 
\end{observation}

\begin{lemma}
There are $\Theta(\log n)$ stages that occur throughout TimSort's main loop. 
\end{lemma}
\begin{proof}
The first three runs have total length $7n/8$. Then there are $n/8$ remaining elements across all remaining runs. At the end of stage $i$, the last run is of size $2^i$. This occurs for all $i$ such that $1 \leq i < \log(n/8)$ by definition. Then $i$ is at most $\log(n/8) - 1$. Thus, we have $\Theta(\log n)$ stages throughout the algorithm.
\qed
\end{proof}

\begin{corollary}\label{lem:TimSort-stages}
There are $\Theta(\log n)$ stages such that $i > k$.
\end{corollary}
\begin{proof}
Follows from the fact that $k$ is a constant.
\qed
\end{proof}

\begin{lemma}\label{lem:tim-cost}
At all stages $i$ such that $i > k$, applying the walk-back algorithm incurs a
walk-back cost of $\Omega(n)$.
\end{lemma}
\begin{proof}
By \Cref{obs:stage-size}, there are $i+3$ runs in the stack at the beginning of this stage. By the fact that $i > k$, this means that the bottom three runs of size $n/2$, $n/4$, and $n/8$ have been forgotten. 
When the merge sequence terminates, there are four runs in the stack. 
This means that the run of size $n/8$ takes label $R_2$. By \Cref{obs:timsort-validate}, we must know $r_2$ to determine that the merge sequence is indeed finished. The result follows from the fact that $r_2 = n/8 \in \Omega(n)$.
\qed
\end{proof}

By \Cref{lem:TimSort-stages} and \Cref{lem:tim-cost}, the total walk-back cost of TimSort under this input instance is in $\Omega(n\log n)$. As described, however, our counterexample input instance has entropy $\Omega(\log n)$. Thus, standard TimSort would also run in time $\Omega(n\log n)$. 

A small modification to our counterexample gives
us the desired result.
Instead of $n/16$ runs of size 2, reduce this to $n/(16\log n)$ runs. Have the remaining elements form a single run of size $n/8 - n/(8\log n)$. See \Cref{fig:timsort-ctr}. 
The Shannon entropy of this new array is constant, so 
standard TimSort would run in linear time.
However, a similar analysis of this variant shows
that applying the walk-back algorithm incurs a walk-back
cost of $\Omega(n\log(n/\log n)) \subseteq
\Omega(n\log n)$.
We conclude the following.

\begin{theorem}\label{thm:timsort-proof}
TimSort is not walkable.
\end{theorem}

\subsection{Additional Counterexample}\label{sec:timsort-proof}

We can use a similar argument to above to show that $\alpha$-MergeSort
is not walkable.

First, let us consider the merge policy of $\alpha$-MergeSort.

\begin{algorithm}
	\caption{$\alpha$-MergeSort's Merge Policy (see \Cref{line:merge-policy-loop} of \Cref{alg:generic-stack})}\label{alg:alpha-mergesort}
	\begin{algorithmic}[1]
        \While{$r_2 < \alpha r_1$ or $r_3 < \alpha r_2$}
            \If{$r_3 < r_1$}
                \State $R_2 \gets \Call{Merge}{R_3, R_2}$
            \Else
                \State $R_1 \gets \Call{Merge}{R_2, R_1}$
            \EndIf
        \EndWhile
    \end{algorithmic}
\end{algorithm}

\begin{theorem}
    There exist input sequences that can be sorted in $O(n)$ time using standard
    $\alpha$-MergeSort that take time $\Omega(n\log n)$ when using $\alpha$-MergeSort
    implemented with the walk-back algorithm. Therefore, $\alpha$-MergeSort is not walkable.
\end{theorem}

\begin{proof}
    Our proof is similar to the proof for TimSort, and in fact, when $\alpha =
    2$, the identical counterexample works.
    Similar sequences work for other values of $\alpha$ as well where the three
    original runs are of sizes $n / \alpha, n / \alpha^2, n / \alpha^3$.
    Without loss of generality, we assume $\alpha = 2$ for this proof and use
    the same counterexample as for TimSort.
    
    We note that, similar to TimSort, only the first while loop condition (on
    line 1) will ever succeed, while the check in the if statement (on line 2)
    will always fail, and consequently only the merge on line 5, between runs
    $R_3$ and $R_2$, will ever occur.
    Therefore the algorithm will perform identical merges to TimSort on this
    input.
    And as with TimSort, the problematic walk-back cost occurs after a stage
    has collapsed, and the stack has just four runs: $\{n / 2, n / 4, n / 8, 2^i\}$.
    On the second while loop condition, where we check if $r_3 < 2 r_2$, when
    implementing the walk-back algorithm, we must at least fully load $R_3$ to
    confirm that the condition is false.
    Therefore, the walk-back cost is at least $n / 8$ (and will in-fact be
    $3 n / 8$), and our TimSort analysis holds.
    \qed
\end{proof}

\section{Additional Walkable Mergesort Proofs}\label{sec:additional-proofs}

In this appendix, we provide proofs showing that additional
mergesorts are walkable.

\subsection{The Original TimSort}

We showed in \Cref{sec:counter} that the modern TimSort algorithm~\cite{auger2019,de2015openjdk} is not walkable. Surprisingly, in this section we show that the original TimSort algorithm~\cite{tim}, in which the fourth condition is removed, is walkable. Refer to \Cref{alg:og-timsort} for the merge policy of this algorithm. 
In this section, we prove the following theorem.

\begin{theorem}
Original TimSort is walkable.
\end{theorem}

\begin{algorithm}
	\caption{Original TimSort's Merge Policy (see \Cref{line:merge-policy-loop} of \Cref{alg:generic-stack})}\label{alg:og-timsort}
	\begin{algorithmic}[1]
        \While{true}
            \If{$|S| > 3$ and $r_1 > r_3$} \Comment{merge condition \#1}
                \State $R_2 \gets \Call{Merge}{R_3, R_2}$
            \ElsIf{$|S| > 2$ and $r_1 \geq r_2$} \Comment{merge condition \#2}
                \State $R_1 \gets \Call{Merge}{R_2, R_1}$
            \ElsIf{$|S| > 3$ and $r_1 + r_2 \geq r_3$} \Comment{merge condition \#3}
                \State $R_1 \gets \Call{Merge}{R_2, R_1}$
            \Else
                \State \textbf{break}
            \EndIf
        \EndWhile
    \end{algorithmic}
\end{algorithm}



As before, we begin by analyzing the walk-back cost of successful merges.

\begin{lemma}
When we perform a merge, the walk-back cost is paid for by at most 2 times the subsequent merge cost. 
\end{lemma}
\begin{proof}
We prove this by discussing the walk-back and merge costs of each of the three conditions. 
By \cref{obs:always-know}, we always know $r_1$ throughout the algorithm. For simplicity, assume that we begin each iteration of
a merge sequence only knowing $r_1$.
To test merge condition 1, we must spend $r_2$ time just walking from the end of $R_1$ to the end of $R_2$ as we assume we do not know $r_2$. If the condition is true, we spend an additional $r_3$ time walking back. As such, the resulting merge of $R_2$ and $R_3$ perfectly pays for the walk-back cost.

Now we discuss condition 2. We know that we just tested and failed condition 1. Again we spent $r_2$ time walking to find the end of $R_3$. However, because condition 1 failed, we only spent $r_1$ time walking before stopping. In doing so, we have recovered $r_2$. Thus, condition 2 can be determined without any additional walk-back cost. We spent $r_1 + r_2$ time walking back and end up merging $R_1$ and $R_2$, so the merge again perfectly pays for the walk-back cost.

Lastly, we discuss condition 3. To test and fail conditions 1 and 2 incurs $r_1 + r_2$ walk-back cost. An additional $r_3$ work is needed to confirm that $r_1 + r_2 \geq r_3$. Thus, we spend $r_1 + r_2 + r_3$ time in total walking back. Since the condition succeeded, we know $r_1 + r_2 \geq r_3$, so the walk-back cost is upper-bounded by twice the merge cost of merging $R_1$ and $R_2$.
\end{proof}

Now we study the walk-back cost of failed merges. 

\begin{observation}
Failing all three merge conditions in Original TimSort incurs a walk-back cost of at most $2(r_1+r_2)$.
\end{observation}
\begin{proof}
We once again assume that we are only aware of $r_1$. As discussed before, failing on condition 1 requires $r_1 + r_2$ time spent walking back and failing on condition 2 incurs no walk-back cost. Finally, failing on condition 3 requires an additional $r_1 + r_2$ time to walk back on $R_3$. 
\qed
\end{proof}

\begin{observation}
Each unique run can contribute to failing all merge conditions at most once as $R_1$ and at most once as $R_2$
\end{observation}
\begin{proof}
If the run was $R_1$, the new run added in the next iteration pushes it to $R_2$. The same run cannot move back to $R_1$ without merging first.
Likewise, if the run was $R_2$, the new run pushes it to $R_3$. It cannot move back to either $R_2$ or $R_1$ without merging.
\qed
\end{proof}

As a result, each run contributes at most 4 times their size (2 times their size as $R_1$ and 2 times their size as $R_2$). This can happen for every original and intermediate run, thus the total walk-back cost of failed merges is $O(m+n)$. Therefore, Original TimSort is walkable.

\subsection{The 2-MergeSort}

In 2019 Buss and Knop introduced several new natural stack-based mergesorts,
which they call the 2-MergeSort and the $\alpha$-MergeSort~\cite{buss19}.
In \Cref{sec:timsort-proof}, we show that $\alpha$-MergeSort (not to be confused with the $\alpha$-StackSort
introduced by Auger et al.~\cite{auger2015}) is not walkable. 
Here, we show that the 2-MergeSort
is.
The merge policy of the 2-MergeSort is described in \Cref{alg:2-merge-sort}.

\begin{algorithm}[h!]
	\caption{2-MergeSort's Merge Policy (see \Cref{line:merge-policy-loop} of \Cref{alg:generic-stack})}\label{alg:2-merge-sort}
	\begin{algorithmic}[1]
		\While{$|S| \geq 3$ and $r_2 < 2 r_1$}
			\If{$r_3 < r_1$}
				\State $R_2 \gets \Call{Merge}{R_3, R_2}$ \Comment{removes
				$R_3$ from $S$}
			\Else
				\State $R_1 \gets \Call{Merge}{R_2, R_1}$ \Comment{removes
				$R_2$ from $S$}
			\EndIf
		\EndWhile
	\end{algorithmic}
\end{algorithm}


\begin{theorem}\label{thm:2-merge-sort}
	2-MergeSort is walkable.
\end{theorem}

\begin{proof}
    As always, we begin by discussing the walk-back costs when a merge occurs.
    By \Cref{obs:always-know}, we always know $r_1$.
    If a merge occurs, then the condition $r_2 < 2r_1$ is true and we incur walk-back cost $r_2$ to test this. If $r_3 < r_1$, we incur additional walk-back cost $r_3$. Since we merge $R_3$ and $R_2$, the merge perfectly pays for the walk-back cost of $r_2 + r_3$. The case is symmetric for if $r_3 \not < r_1$.

    Now we discuss the cost associated when no merge occurs. This occurs when we fail the initial comparison, so $r_2 \geq 2r_1$. Therefore, we incur walk-back cost $2r_1$. We observe that this can only happen once per unique run during the main loop. Whenever a failure occurs, we push a new run, displacing $R_1$ to $R_2$. To return to spot one of the stack, this run will have to merge, destroying it. Thus, the total walk-back cost is $O(m+n)$.
	\qed
\end{proof}

\section{In-Place Run Partitioning} \label{sec:compression}
\label{app:in-place}
In this appendix, for the sake of completeness, we describe a
simple partitioning algorithm
that can be used, for example, to move all short runs in an array, $A$,
to the end of $A$, which also moves the long runs to be before
that end region, without destroying the long runs. 

We can abstract this problem by representing each element of a long run as $0$'s,
and each element of a short run as $1$'s.
Suppose $A$ is an array of $n$ 
elements, each of which is either $0$ or $1$, such that there are $n-m$
0's and $m$ 1's in $A$.
Our goal of moving the short runs to the end of $A$ is equivalent to segregating
the 0's and 1's in $A$, such that the $n-m$ 0's precede the $m$ 1's.
While it is important for us that the 0's are moved stably, i.e., that the long
runs are not destroyed, it is not important for us and in general our algorithm
will not preserve the stability of the 1's, i.e., the short runs.
This can be done in-place in $O(n)$ time by making a single pass through $A$
while maintaining two pointers, one for the next 0 and one for the first
preceding 1.
See \Cref{alg:compress}.

In our setting, we do not know which elements are $0$'s or $1$'s {\it a priori}.
However, since these labels are determined by the length of each run and our 
algorithm performs a single pass through $A$, we can 
scan each run as we come across them to determine whether they are a run of $0$'s
or $1$'s. We may scan each run at most twice, once for each pointer, 
so this adds an additional $O(n)$ time.

\begin{algorithm}
	\caption{In-Place Run Partitioning}\label{alg:compress}
	\begin{algorithmic}[1]
		\State Let $A[i]$ be the first cell with $A[i]=1$
		\State Let $A[j]$ be the first cell with $A[j]=0$ and $j>i$
		\State Swap $A[i]$ and $A[j]$
		\While{$i<n$} 
			\State Let $A[i]$ be the next cell with $A[i]=1$
			\State Let $A[j]$ be the next cell with $A[j]=0$ and $j>i$
			\State Swap $A[i]$ and $A[j]$
		\EndWhile
	\end{algorithmic}
\end{algorithm}

Note that this algorithm will move 0's forward in a way that
preserves their original order in the array, $A$.
Thus, we have the following:

\begin{theorem}
Given an array, $A$, of $n$ cells, $m$ of which store 1's
and $n-m$ of which store 0's, we can stably partition
the 0's in-place to the first $n-m$ cells in $A$ in $O(n)$ time.
\end{theorem}

\section{Experimentally Evaluating Walk-Back Cost}\label{sec:experiments}

\begin{figure}[t]
	\centering
	\includegraphics[width=0.9\linewidth]{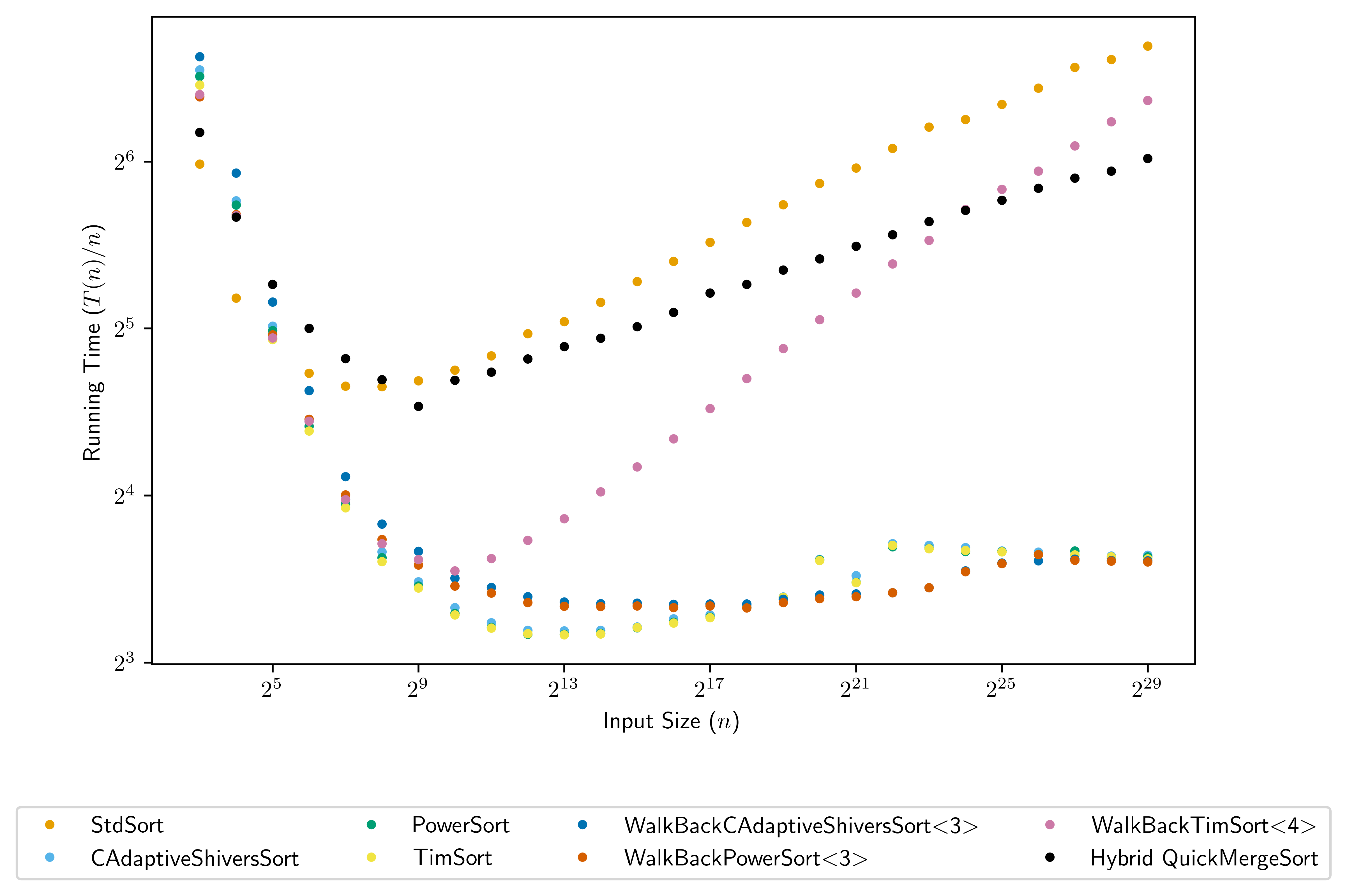}
	\caption{
		A comparison of normalized running times on our TimSort counterexample
		inputs as described in \Cref{sec:counter}. 
		Recall that for this input the Shannon entropy is constant.
		Thus, we scale the running times by $n$.
		As opposed to our stack-based mergesorts, \texttt{std::sort} and Hybrid
		QMS are not instance-optimal, so they cannot take advantage of the
		natural ordering of the input.
		Further, while all our in-place versions of the \textit{walkable}
		mergesorts are able to take advantage of the natural ordering of the
		input, walk-back TimSort is not.
        This confirms our proof in \Cref{sec:counter} that TimSort is not walkable 
        and it requires the encoding techniques described in \Cref{sec:encoding}.
	}
	\label{fig:running-times-bad-input}
\end{figure}

We support our theoretical results with experiments by comparing the
performance of several stack-based mergesorts with their in-place versions
produced from the walk-back algorithm. 
We also compare our algorithms against the performance of the built-in
C++ implementation of introsort, \texttt{std::sort}, and against QuickXsort, a
recent in-place sorting algorithm.

\subsection{Experimental setup}

We ran our C++ experiments on 
\texttt{Ubuntu 22.04.5
LTS} (Linux kernel\\ \texttt{5.15.0-135-generic}). The hardware comprised two
Intel\textregistered{} Xeon\textregistered{} X5680 processors (@ 3.33\,GHz,
24\,MiB total L3 cache) and 94\,GiB of RAM. The code was compiled using
\texttt{GCC} version \texttt{11.4.0} with the \texttt{-O3 -march=native}
optimization flags.
We ran our benchmarks using \texttt{Google Benchmark} version \texttt{1.8.3},
which is a microbenchmarking framework for C++.
We tested two types of inputs: (1) our TimSort counterexample input as described in
\Cref{sec:counter} 
and (2) a uniform input distribution formed by shuffling distinct integers
using the Mersenne Twister 19937 random number generator.
Many QuickXsort variants were introduced by Edelkamp, Wei{\ss}, and Wild
\cite{edelkamp2020quickxsort}, and by Edelkamp and Wei{\ss}
\cite{edelkamp2019worst}.
We picked the Hybrid QuickMergeSort (Hybrid QMS) variant, which is the most
competitive variant described in \cite{edelkamp2019worst}.
We implemented TimSort as described in Algorithm 3 of~\cite{auger2019}.

For simplicity, we used \texttt{std::merge} from the C++ standard library for
all the canonical stack based mergesorts, and \texttt{std::inplace\_merge} for
the in-place versions.
While the former is guaranteed to run in $O(n)$ time, C++'s in-place merge
algorithm has a worst case time complexity of $O(n\log n)$, although our
experiments suggest that this worst-case condition was not reached.
Another detail is that none of our inputs allowed duplicate elements, although
preliminary experiments suggest that this does not have a significant impact on
the running time of the algorithms.
All our code will be made available online.

\subsection{Results}

\begin{figure}[t]
	\centering
	\includegraphics[width=0.9\linewidth]{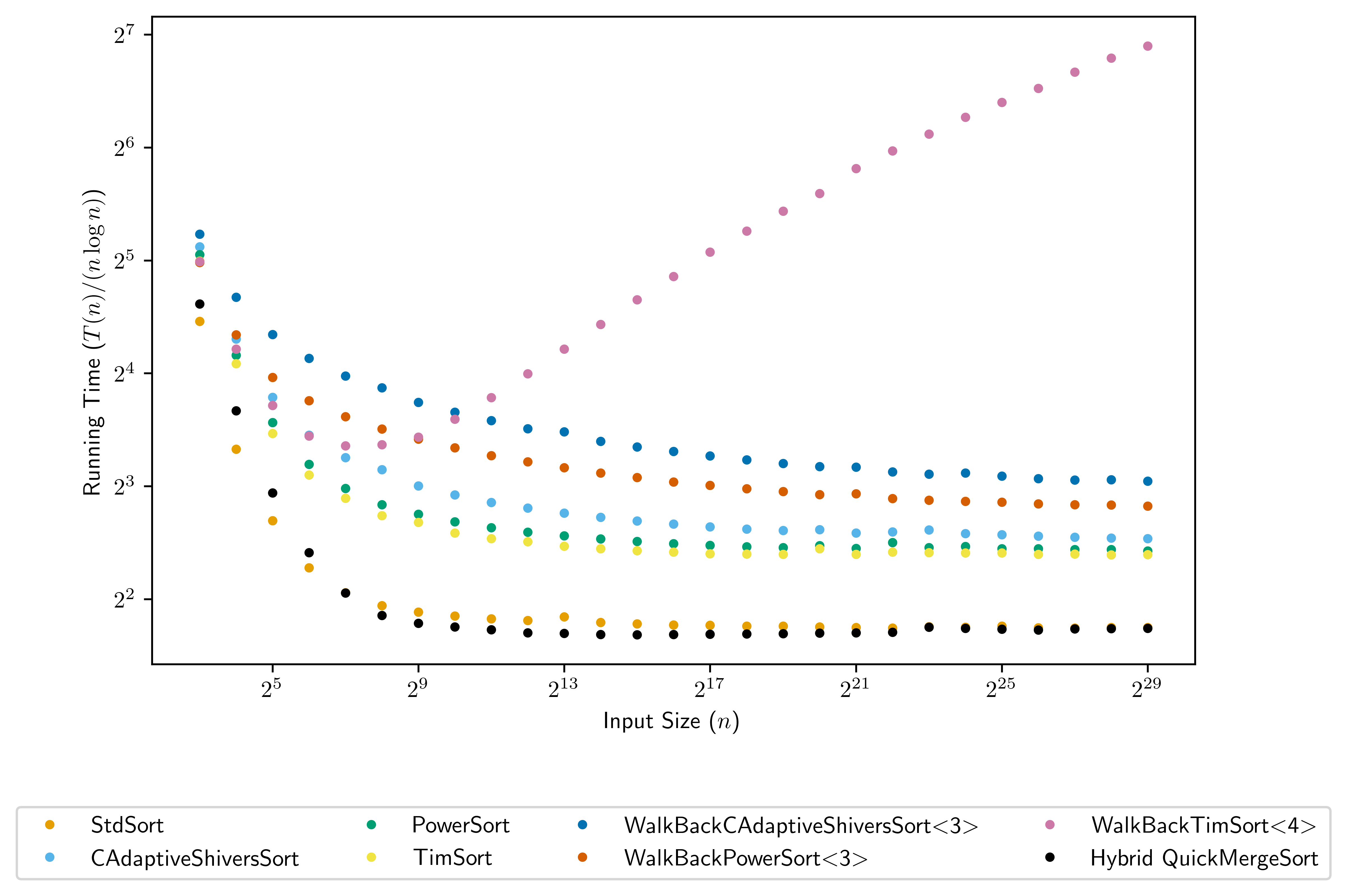}
	\caption{
		A comparison of normalized running times on uniformly random inputs.
		Since uniformly random inputs are expected to have a Shannon entropy of
		$\Theta(\log{n})$, we scaled the running times by $n \log{n}$, such that
		any algorithm that is optimal with respect to Shannon entropy should
		appear constant.
		As expected, \texttt{std::sort} (StdSort) and Hybrid QMS are optimized
		well for these inputs.
		Reassuringly, not only do all the stack-based mergesorts appear
		constant, but also all our in-place versions of the \textit{walkable}
		mergesorts appear constant.
		The only mergesort which does not appear constant is the in-place
		version of TimSort, which is not walkable, as shown in \Cref{sec:counter}.
	}
	\label{fig:running-times-normal}
\end{figure}

In \Cref{fig:running-times-bad-input}, we compare the algorithms on the TimSort counterexample
input as described in \Cref{sec:counter} and \Cref{sec:timsort-proof}.
This input has a constant Shannon entropy, so any algorithm that is
instance-optimal with respect to Shannon entropy should run in linear time, and
after scaling the running times by $n$, should appear constant.
Indeed, we see that all the standard stack-based mergesorts take advantage of the
natural ordering of the input and appear constant, while C++'s
\texttt{std::sort} function (using introsort) and QuickXsort, which are not
instance-optimal, perform poorly.
Reassuringly, after applying our simple walk-back algorithm, the in-place
versions of the \textit{walkable} mergesorts (PowerSort and $c$-Adaptive
ShiversSort) appear to remain input-sensitive like their non-in-place counterparts.
In contrast, the in-place version of TimSort, which as we proved in
\Cref{sec:counter} is not walkable, performs poorly.

In \Cref{fig:running-times-normal}, we compare the algorithms on uniformly random inputs,
which are expected to have a Shannon entropy of $\Theta(\log{n})$, so we scale
the running times by $n\log{n}$, such that any $O(n\log{n})$ algorithm should
appear constant.
And indeed, all the algorithms (notably except the version of TimSort implemented with the walk-back algorithm)
appear constant, as expected.
The fact that walk-back TimSort appears to perform worse than
$O(n\log{n})$ on such inputs is interesting.



\section{Supplemental Experiments} \label{app:supplemental-experiments}


\begin{figure}[t]
	\centering
	\begin{subfigure}[b]{0.48\textwidth}
		\centering
		\includegraphics[width=\textwidth]{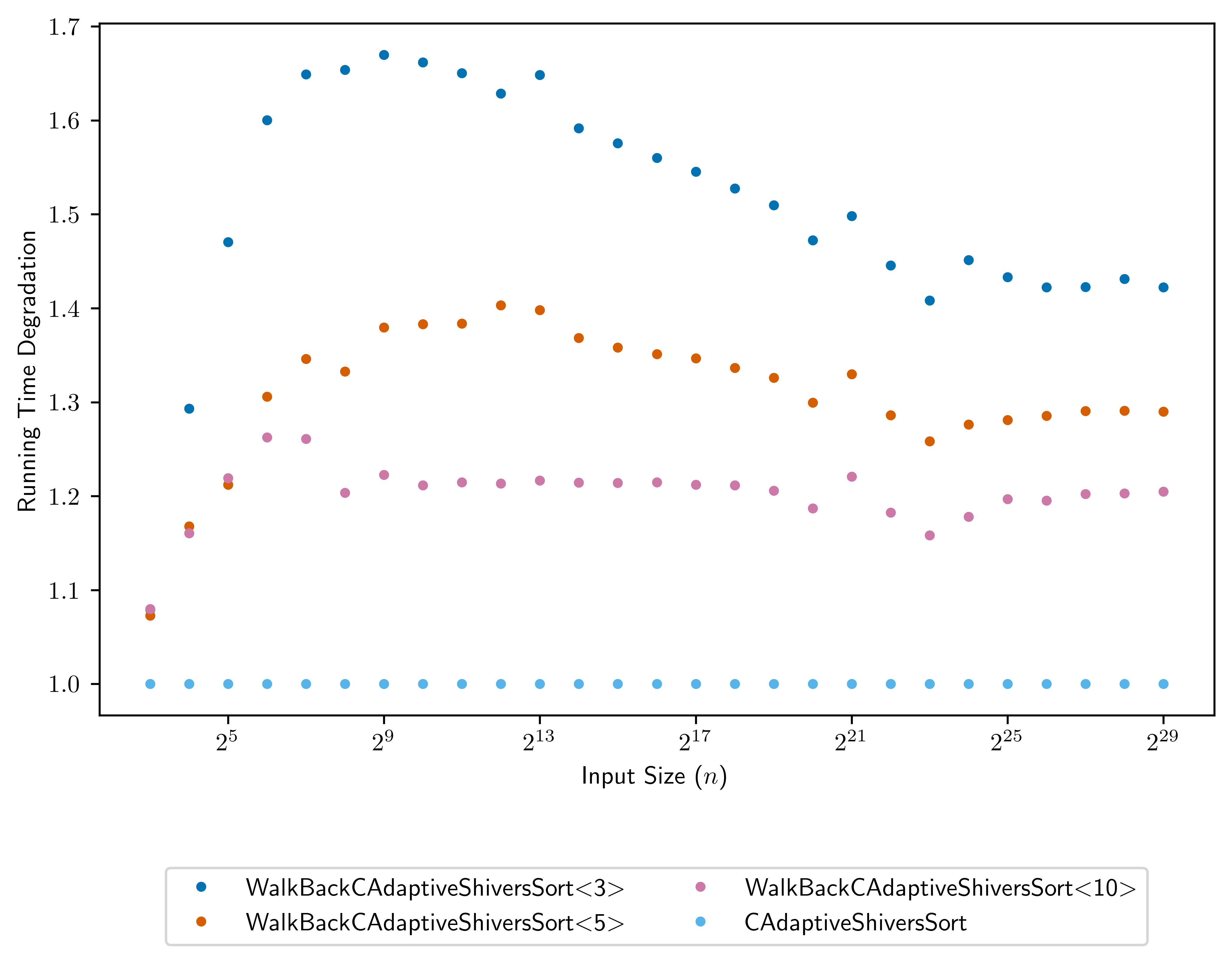}
	\end{subfigure}
	\hfill
	\begin{subfigure}[b]{0.48\textwidth}
		\centering
		\includegraphics[width=\textwidth]{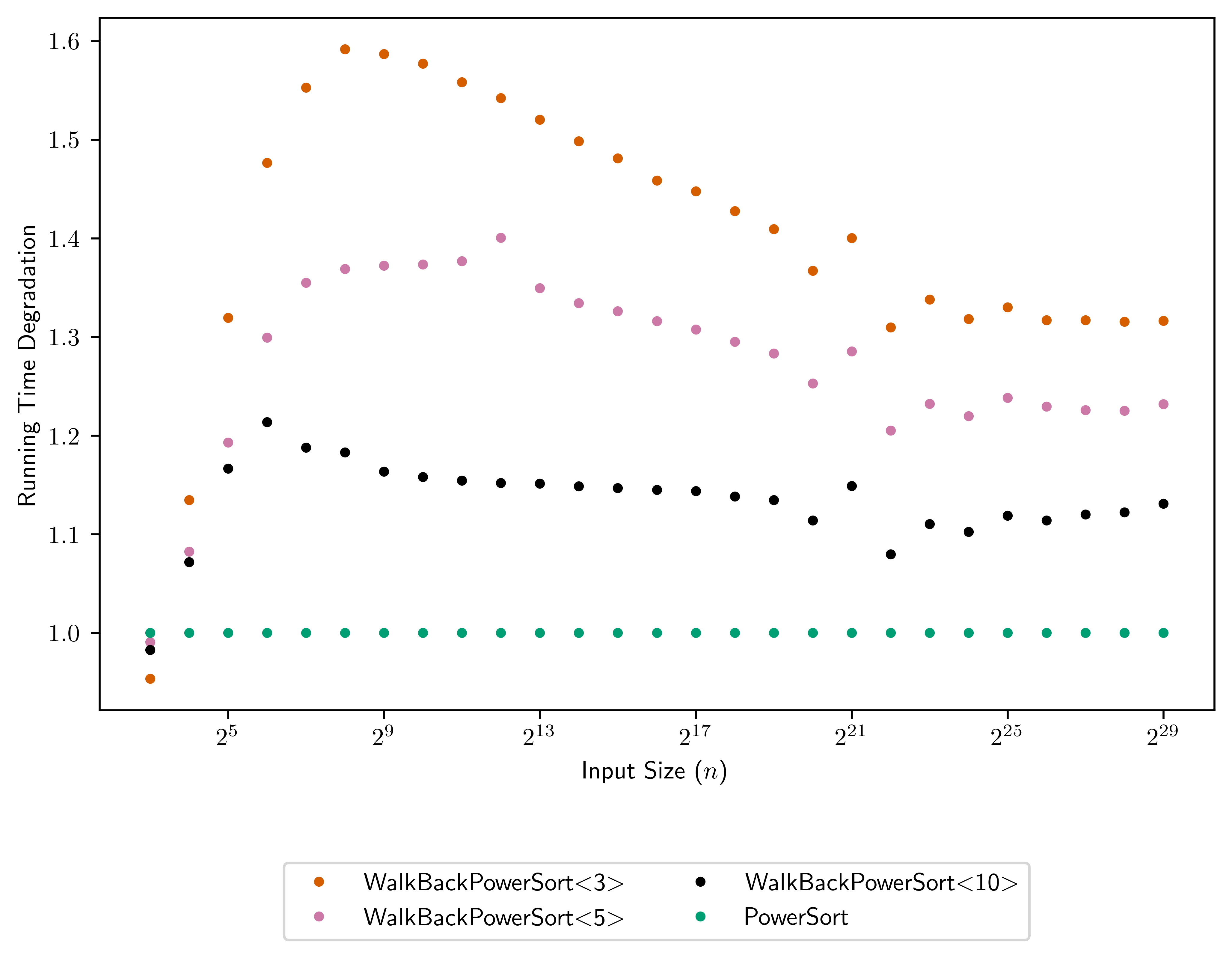}
	\end{subfigure}
	\vspace{1em}
	\begin{subfigure}[b]{0.48\textwidth}
		\centering
		\includegraphics[width=\textwidth]{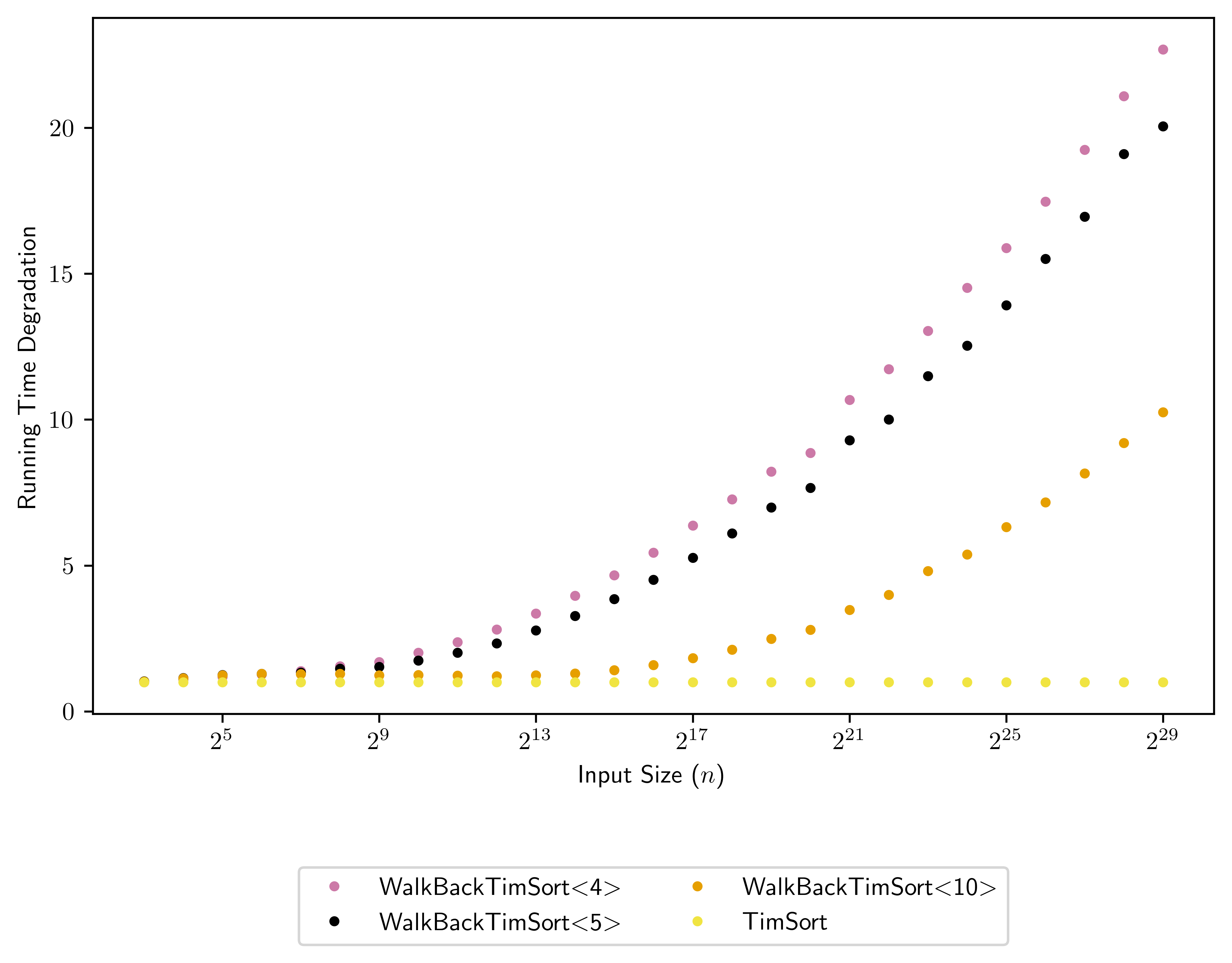}
	\end{subfigure}
	\caption{
		A comparison of the different in-place versions of each algorithm
		against their standard version on uniformly random inputs.
		Each algorithm is measured with respect to the standard version, which
		is normalized to 1.
		While the in-place versions of $c$-Adaptive ShiversSort and PowerSort
		are within a small constant factor of the standard version, the in-place
		versions of TimSort clearly perform asymptotically worse.
	}
	\label{fig:all-compare}
\end{figure}

\subsection{Impact of Different Stack Sizes} \label{sec:stack-size}

In the previous section, we gave each in-place variant their minimal stack size
required, $k$, which for PowerSort and $c$-Adaptive ShiversSort is 3, and for
TimSort is 4.
We now compare each in-place algorithm as its stack size increases
from $k$ to 5, to 10, against the standard version (which can be thought of as
having stacks of size up to $\Theta(\log n)$).

As expected, all the in-place versions of our stack-based natural mergesorts
perform better as their stack size increases, performing the best when the stack
size is not bounded.
The walkable mergesorts in particular perform very well, performing within a very
small constant factor of the standard version.
Even for the smallest stack size of just 3, the in-place versions of PowerSort
and $c$-Adaptive ShiversSort perform within a factor of 1.7 and 1.6,
respectively, of the standard version.
This further confirms that not only is our simple walk-back algorithm effective,
but that it is likely benefiting from being very cache-friendly.
In contrast, despite the cache-friendliness of our walk-back algorithm, the
in-place version of TimSort clearly does not keep up with the standard version.
Even when allowed to store the top 10 runs, it is only able to keep up with the
standard version for so long before diverging.
See \Cref{fig:all-compare}.
This directly aligns with our results from \Cref{sec:counter}.

\subsection{Stack Height Variability} \label{sec:stack-height-variability}

\begin{figure}[t]
	\centering
	\begin{subfigure}[b]{0.48\textwidth}
		\centering
		\includegraphics[width=\textwidth]{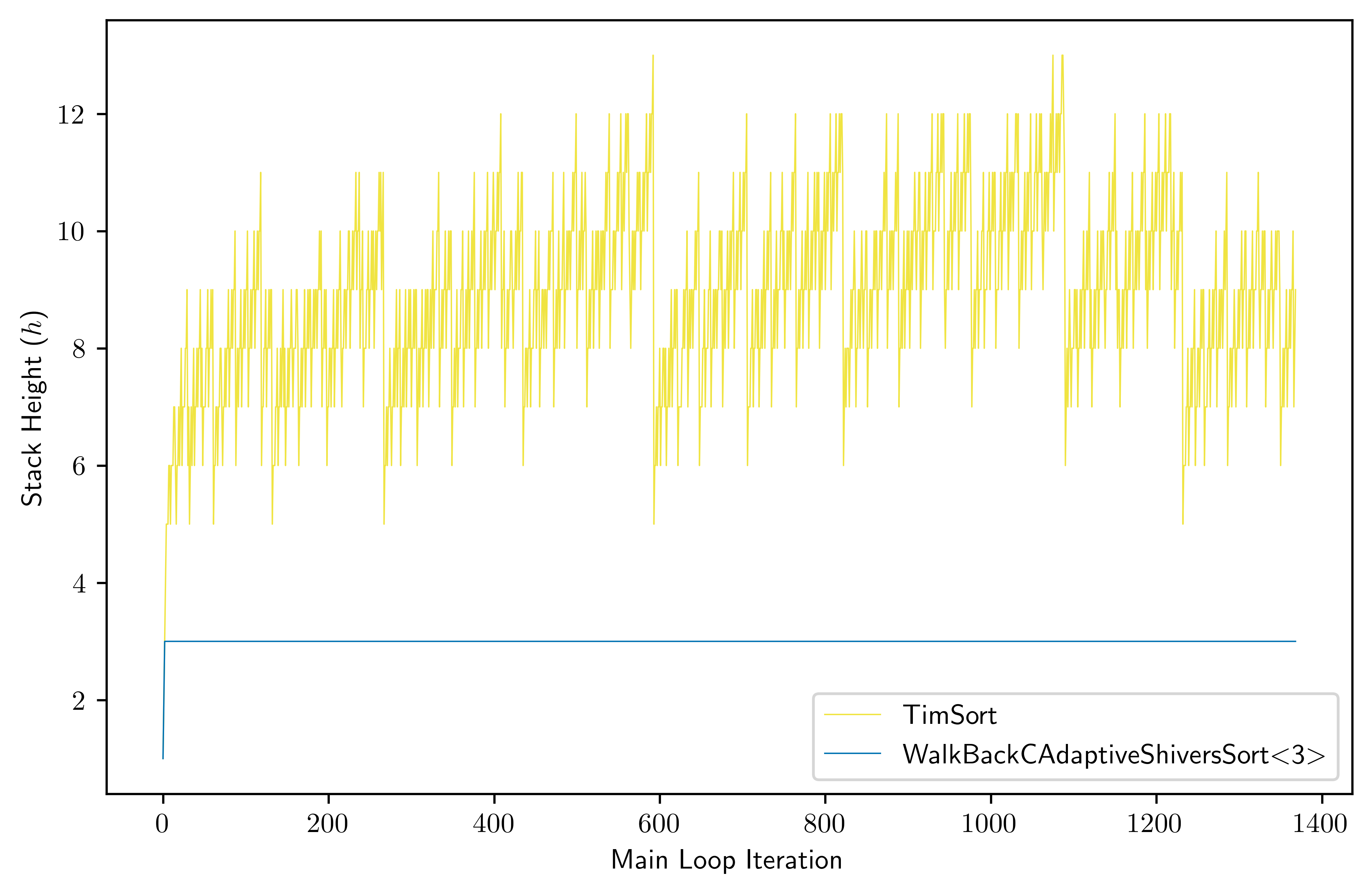}
	\end{subfigure}
	\hfill
	\begin{subfigure}[b]{0.48\textwidth}
		\centering
		\includegraphics[width=\textwidth]{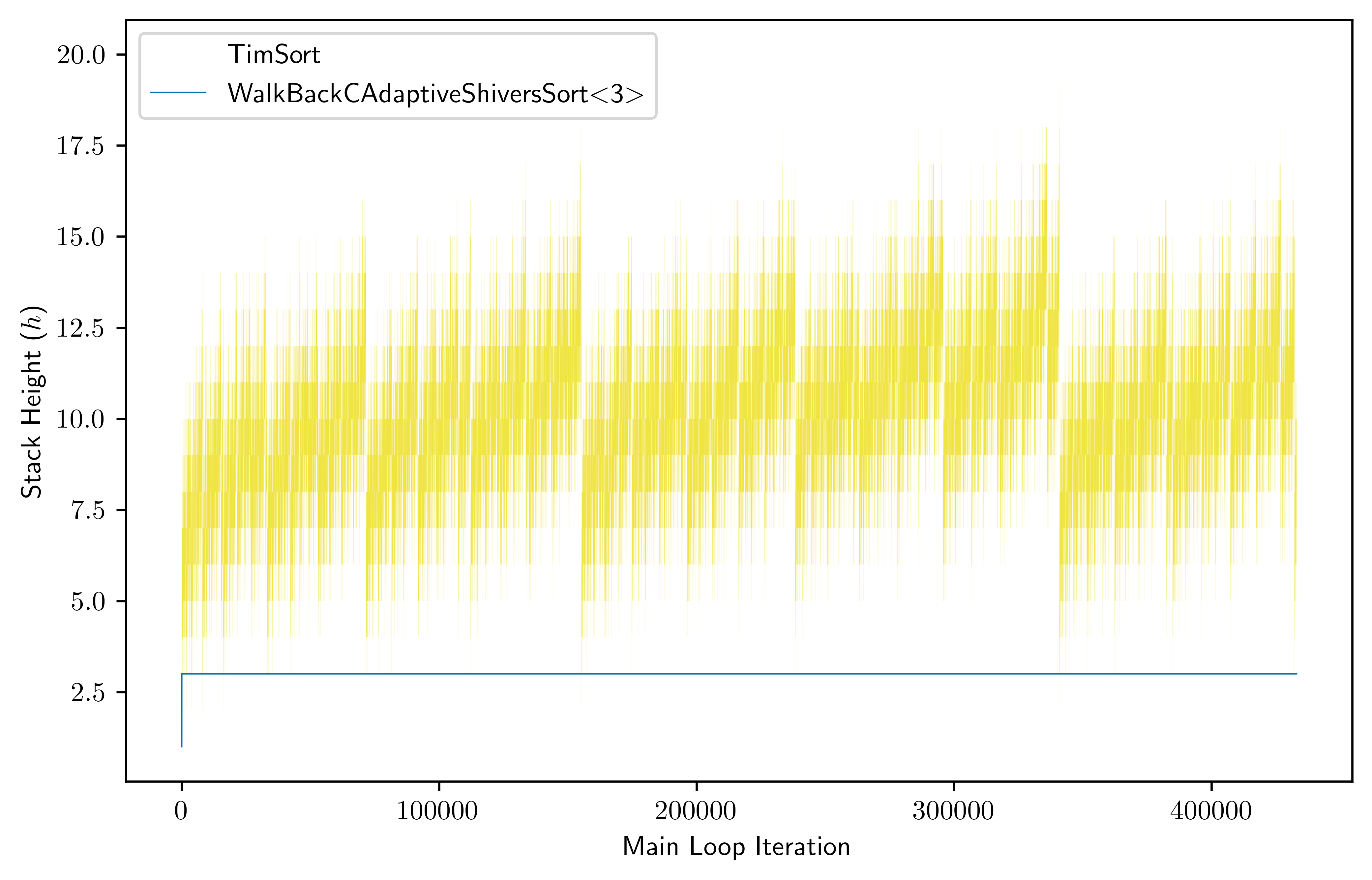}
	\end{subfigure}
	\caption{
		Here we track TimSort's stack height ($h$) during each iteration of the main
		loop, right after a new run is pushed onto the stack $S$ for arrays of
		size $n = 2^{20}$.
		On the left, we use our TimSort counterexample input as described in
		\Cref{sec:counter} while on the right we use uniformly random inputs.
	}
	\label{fig:bad-tim-sort-stack-sizes}
\end{figure}

In \Cref{fig:bad-tim-sort-stack-sizes}, we plot TimSort's stack height ($h$)
during each iteration of the main loop, right after a new run is pushed onto the
stack $S$ for arrays of size $n = 2^{20}$ for both our bad TimSort example
(left) and uniformly random inputs (right), and compare it to the stack size of
the in-place version of $c$-Adaptive ShiversSort.
It is clear to see from both plots how TimSort's stack size often significantly
fluctuates in hard-to-predict ways, while the in-place version of $c$-Adaptive
ShiversSort has a very consistent stack size, bounded by 3.
These simple figures demonstrate the predictability of our in-place versions,
which may lead to their implementations being less error-prone.

%


\end{document}